\documentclass[preprint]{imsart}
\pdfoutput=1
\RequirePackage[OT1]{fontenc}
\usepackage{amsthm,amsmath,natbib,amssymb,bm,enumerate}
\RequirePackage[colorlinks,citecolor=blue,urlcolor=blue]{hyperref}
\usepackage{booktabs}

\startlocaldefs
\numberwithin{equation}{section}
\theoremstyle{plain}
\newtheorem{thm}{Theorem}[section]
\newtheorem{lemma}{Lemma}[section]
\newtheorem{corollary}{Corollary}[section]
\theoremstyle{remark}
\newtheorem{remark}{Remark}[section]

\endlocaldefs
\newcommand{\plim}{\mathop{\mathrm{plim}}}
\newcommand{\amax}{\mathop{\mathrm{argmax}}}
\newcommand{\bgm}{\gamma}
\newcommand{\bep}{\bm{\epsilon}}
\newcommand{\bbe}{\bm{\beta}}
\newcommand{\BF}{\mathrm{BF}}
\newcommand{\gvf}{_{\gamma:F}}
\newcommand{\BIC}{[\mathrm{BIC}]}
\newcommand{\mcM}{\mathcal{M}}
\newcommand{\PR}{\mathrm{Pr}}
\newcommand{\sss}{SS }
\newcommand{\pigm}{\pi}
\newcommand{\pial}{\pi}
\newcommand{\piet}{\pi}
\newcommand{\pibeet}{\pi}
\newcommand{\pig}{\pi}

\def\citeapos#1{\citeauthor{#1}'s (\citeyear{#1})}

\allowdisplaybreaks

\begin{document}

\begin{frontmatter}
\title{Robust Bayesian variable selection with sub-harmonic priors}
\runtitle{Bayes factors}

\begin{aug}
\author{\fnms{Yuzo} \snm{Maruyama}
\thanksref{t1,m1}\ead[label=e1]{maruyama@csis.u-tokyo.ac.jp}}
\and
\author{\fnms{William, E.} \snm{Strawderman}
\thanksref{t2,m2}\ead[label=e2]{straw@stat.rutgers.edu}}
\thankstext{t1}{This work was partially supported by KAKENHI \#21740065 \& \#23740067.}
\thankstext{t2}{This work was partially supported by a grant from the Simons
Foundation (\#209035 to William Strawderman).}

\address{University of Tokyo\thanksmark{m1} and Rutgers University\thanksmark{m2} \\
\printead{e1,e2}}

\runauthor{Y. Maruyama and W. Strawderman}

\end{aug}

\begin{abstract}
This paper studies Bayesian variable selection in linear models with general spherically symmetric error distributions. 
We propose sub-harmonic priors which arise as a class of mixtures of Zellner's $g$-priors for which the Bayes factors are independent of the underlying error distribution, as long as it is in the spherically symmetric class. 
Because of this invariance to spherically symmetric error distribution, we refer to our method as a robust Bayesian variable selection method.
We demonstrate that our Bayes factors have model selection consistency and are coherent. 
We also develop Laplace approximations to Bayes factors for a number of recently studied mixtures of $g$-priors that have recently appeared in the literature (including our own) for Gaussian errors.
These approximations, in each case, are given by the Gaussian Bayes factor based on BIC times a simple rational function of the prior's hyper-parameters and the $R^2$'s for the respective models. 
We also extend model selection consistency for several $g$-prior based Bayes factor methods for Gaussian errors to the entire class of spherically symmetric error distributions.
Additionally we demonstrate that our class of sub-harmonic priors are the only ones within a large class of mixtures of $g$-priors studied in the literature which are robust in our sense.
A simulation study and an analysis of two real data sets indicates good performance of our robust Bayes factors relative to BIC and to other mixture of $g$-prior based methods.
\end{abstract}

\begin{keyword}[class=AMS]
\kwd[Primary ]{62F15}
\kwd{62F07}
\kwd[; secondary ]{62A10}
\end{keyword}

\begin{keyword}
\kwd{Bayes factor}
\kwd{Bayesian variable selection}
\kwd{fully Bayes method}
\kwd{model selection consistency}
\kwd{sub-harmonic prior}
\end{keyword}
\end{frontmatter}

\section{Introduction}
\label{sec:intro}
Suppose the linear regression model is used to relate $Y$ to the
$p$ potential predictors $x_1, \dots, x_p$,
\begin{equation} \label{full-model}
\bm{y} = \alpha \bm{1}_n+\bm{X}_F\bbe_F + \sigma_F\bep_F ,
\end{equation}
 where the subscript $F$ refers to the full model $\mcM_F$.
In the model \eqref{full-model}, $\alpha$ is an unknown intercept parameter,
$\bm{1}_n$ is an $n\times 1$ vector of ones,
$\bm{X}_F=(\bm{x}_1,\dots, \bm{x}_p)$ is an $n \times p$ design matrix,
and $\bbe_F$ is a $p \times 1$ vector of unknown
regression coefficients.
In the error term of \eqref{full-model}, $\sigma_F$ is an unknown scalar
and $\bep_F $ has a spherically symmetric (SS) distribution with the density
$f_n(\|\bep_F\|^2)$, $E[\bep_F]=\bm{0}_n$ and $\mathrm{Var}[\bep_F]=\bm{I}_n$.
We assume that the columns of $\bm{X}_F$ have been standardized so that for 
$1 \leq i \leq p$, $\bm{x}'_i\bm{1}_n = 0$ without loss of generality.

We shall be particularly interested in the variable selection 
problem where we would like to select an unknown
subset of the effective predictors.
It will be convenient throughout to index each of these $2^p$ possible
subset choices by the vector
\begin{equation*}
\bm{\gamma} =(\gamma_1, \dots,\gamma_p)'
\end{equation*}
where $\gamma_i=0$ or $1$. 
We use $q_{\bgm}=\bm{\gamma}'\bm{1}_p$ to denote the size of
the $\bm{\gamma}$th subset.
The problem then becomes that of selecting a submodel of \eqref{full-model}
\begin{equation}\label{submodel-gamma}
\bm{y} = \alpha \bm{1}_n+\bm{X}_{\bgm} \bbe_{\bgm} + \sigma_{\bgm}\bep_{\bgm}.
\end{equation}
In \eqref{submodel-gamma}, $\bm{X}_{\bgm}$ is the $n \times q_{\bgm}$ 
matrix whose columns 
correspond to the $\bm{\gamma}$th subset of $x_1,\dots,x_p$, $\bbe_{\bgm}$
is a $q_{\bgm} \times 1$ vector of unknown regression coefficients.
Let $\mcM_{\bgm}$ denote the submodel given by \eqref{submodel-gamma}. 
We assume the error term $\bep_{\bgm} $ has 
the \sss density
\begin{equation}\label{bep_sim_f}
 \bep_{\bgm}\sim f_n(\|\bep_\bgm\|^2)
\end{equation}
for all $\bm{\gamma}$, 
with $E[\bep_{\bgm}]=\bm{0}_n$ and $\mathrm{Var}[\bep_{\bgm}]=\bm{I}_n$.
Further $\sigma_{\bgm}$ is an unknown scalar in the error term.
We note that, in most earlier studies, the error terms 
in linear models have been assumed to have a Gaussian distribution, 
e.g.~as in \cite{George-Foster-2000} and \cite{Liang-etal-2008}. 
There are a number of reasons to broaden the class of error distributions
from the Gaussian to the general \sss family, chief among them being
that it allows error distributions with flatter tails than
the normal. Among the commonly assumed distributions, the class of
multivariate-$t$ distributions given by
\begin{equation*}
 f_n(\|\bep_\bgm\|^2)=\frac{\Gamma(\{l+n\}/2)}{(l\pi)^{n/2}\Gamma(l/2)}
\left(1+\|\bep_\bgm\|^2/l\right)^{-(l+n)/2}
\end{equation*}
with degrees of freedom $l$, is the best known and most useful.
More generally, the class of scale mixtures of normals allows, by 
De Finetti's theorem, a huge class of exchangeable error distributions
with flatter tails than the normal. However the class of \sss distributions
is much broader.
While the class of multivariate-$t$ distributions or scale mixtures of
normals may be the most interesting extensions to the normal class,
our results will hold for the entire class of \sss distributions.
Note that for all such distributions the coordinates are uncorrelated, but,
except for the normal distribution they are dependent.
Standard references for linear models with \sss distributions
include \cite{Anderson-2003} and \cite{Fang-Zhang-1990} among others.

In this paper, we assume that $n > p+1$ (the so called classical setup)
and $\{\bm{x}_1,\dots,\bm{x}_p\}$ are linearly independent, 
which implies that 
\begin{equation}
\mathrm{rank}\ \bm{X}_F=p, 
\quad \mathrm{rank}\ \bm{X}_{\bgm}=q_{\bgm}.
\end{equation}
We also assume in much of the paper 
that the null model $\mcM_N$ ($q_{\bgm}=0$ or $ \bm{\gamma}=(0,\dots,0)'$)
is not a possible model, that is, the number of possible models is $2^p-1$, 
rather than, $2^p$.
Some reviewers of a previous version of this paper have objected on philosophical
grounds to the exclusion of the null model.
While we agree that in some studies a fundamental scientific
issue is to decide if there are any significant effects, in many other
cases that particular issue has been addressed and the issue is to
select those variables that have the most important effects.
Thus for most of the paper we will assume that the null model is not a possible choice.
However, in Section \ref{sec:null}, we give a slightly modified 
development that allows positive prior probability on all submodels
including the null model.
All the results in the paper remain true under this modification.

A Bayesian approach to this problem entails the specification of 
prior distributions on the models $\mathrm{Pr}(\mcM_{\bgm})$, 
and on the parameters  
$\alpha,\bbe_{\bgm}, \sigma_{\bgm}$ of each model.  
For each such specification, 
of key interest is the posterior probability of $\mcM_{\bgm}$
given $\bm{y}$,
\begin{equation} \label{posterior-2}
 \mathrm{Pr}(\mcM_{\bgm}|\bm{y})
=\frac
{\PR(\mcM_{\bgm}) m_{\bgm}(\bm{y})}
{\sum_{\bgm} \PR(\mcM_{\bgm}) m_{\bgm}(\bm{y})}
=\frac
{\PR(\mcM_{\bgm}) \BF\gvf}
{\sum_{\bgm} \PR(\mcM_{\bgm}) \BF\gvf},
\end{equation}
where $\mathrm{Pr}(\mcM_N)=0$ is assumed as mentioned in the above.
In \eqref{posterior-2}, $m_{\bgm}(\bm{y})$ is the marginal density under
$\mcM_{\bgm}$ and 
$\BF\gvf$ is the Bayes factor for comparing
each of $\mcM_{\bgm}$ to the full model $\mcM_F$
which is defined as
\begin{equation*}
 \BF\gvf
= \frac{m_{\bgm}(\bm{y})}{m_F(\bm{y})},
\end{equation*}
where $m_F(\bm{y})$ is the marginal density under the full model.
In Bayesian model selection,
\begin{equation} \label{eq:relation_m_BF}
\amax_{\bgm} \mathrm{Pr}(\mcM_{\bgm}|\bm{y})=\amax_{\bgm} \PR(\mcM_{\bgm}) \BF\gvf,
\end{equation}
is typically selected as the best model.

In this paper, the main focus is on $ \BF\gvf $, not $\PR(\mcM_{\bgm}) $.
Hence our main aim is to propose and study specifications for
the prior distribution of the parameters for
each submodel $\mcM_{\bgm}$.
We have four main goals, three of which are motivated by results
in \cite{Maru-Straw-2005} on minimax generalized Bayes (GB) ridge
regression estimators. \cite{Maru-Straw-2005} give a class of
separable priors that lead to GB estimators which do not depend on
the particular underlying \sss error distribution and are simultaneously
minimax for all such distributions with finite second moment.
Within this class, a subclass was found such that the GB estimators had
simple analytical forms.

One corresponding goal in this paper is to find conditions on the prior
distributions that lead to robustness of the Bayesian variable selection
procedure in the sense that the Bayes factors are independent of the 
particular \sss error distribution, these conditions turn out to be
separability condition as in \cite{Maru-Straw-2005}.

The second is to find, among this separable class of
priors, a subclass for which the corresponding GB estimators in each
submodel are admissible and (at least nearly) minimax over the
entire \sss class. This is achieved by choosing the ``sub-harmonic priors''
on $\bbe_{\bgm}$ described in Section \ref{sec:prior}.
In particular, the joint density we consider for $\mcM_\bgm $ has the form
\begin{equation} \label{intro-prior}
\pigm(\alpha,\bm{\theta}_{\bgm}, \sigma_{\bgm}|\nu)
\propto 
\sigma_{\bgm}^{-\nu-1} \|\bm{\theta}_{\bgm} \|^{-q_{\bgm}+\nu}
\end{equation}
for 
$\bm{\theta}_{\bgm}=(\theta_1,\dots,\theta_{q_{\bgm}})'=(\bm{X}'_{\bgm} \bm{X}_{\bgm})^{1/2}\bbe_{\bgm}$ 
and a non-random hyper-parameter $\nu$ with $0< \nu< q_{\bgm}$. 
Since the term including $\bm{\theta}_{\bgm}$ in the prior above, 
$ \|\bm{\theta}_{\bgm} \|^{-q_{\bgm}+\nu}$ for $0<\nu<\min(2,q_{\bgm})$,
is known as a sub-harmonic function, that is,
\begin{equation*}
\textstyle{\sum_{i=1}^{q_{\bgm}} \left\{\partial^2/\partial \theta_i^2\right\}
\|\bm{\theta}_{\bgm} \|^{-q_{\bgm}+\nu}>0},
\end{equation*}
we call the prior given by \eqref{intro-prior} a sub-harmonic prior.
We will show that such priors lead to robust Bayes factors, in the sense that
each Bayes factor does not depend on the form of the underlying error 
distribution. 
We indicate that the corresponding generalized Bayes estimators for each
submodel are admissible and ``nearly'' minimax.

The third goal is to find a class of priors for which the Bayes factors
have a tractable form. The class of mixtures of $g$-priors is a large
class which leads to tractable Bayes factors in the Gaussian case.
We demonstrate that our sub-harmonic priors are unique among the
class of $g$-mixture priors studied in the literature that are robust in
our sense (see Lemma \ref{lem:uniqueness}).

The final goal is to show that the resulting procedures have model selection
consistency over the entire class of \sss distributions.
This is accomplished by deriving an approximation to the Bayes factors
in terms of BIC based Bayes factors and showing 
consistency of the BIC procedures implies that of the Bayes factors 
based procedures.

The organization of this paper is as follows.
In Section \ref{sec:prior}, we give details of the prior distribution.
In Section \ref{sec:marginal},
we show that the Bayes factor with respect to the above prior 
is given by
\begin{equation}\label{BF:intro}
\BF\gvf(\nu)=\BF^G\gvf(\nu)
\end{equation}
where
\begin{equation} \label{BFG:intro}
\BF^G\gvf(\nu)
=\frac
{\int_0^{\infty} g^{\frac{\nu}{2}-1}(1+g)^{\frac{n-q_{\bgm}-1}{2}}
\{g(1-R_{\bgm}^2)+1 \}^{-\frac{n-1}{2}}\, dg}
{\int_0^{\infty} g^{\frac{\nu}{2}-1}(1+g)^{\frac{n-p-1}{2}}
\{g(1-R_F^2)+1 \}^{-\frac{n-1}{2}}\, dg},
\end{equation}
for $0<\nu<\min_{\bgm} q_{\bgm}=1$. 
In \eqref{BFG:intro}, $\BF^G\gvf(\nu) $ 
is the Bayes factor for standard Gaussian errors and
\begin{equation*}
 R_{\bgm}^2
=1-
\frac{\|\bm{Q}_\bgm(\bm{y}-\bar{y}\bm{1}_n)\|^2}{\|\bm{y}-\bar{y}\bm{1}_n\|^2},\quad
 R_{F}^2
=1-
\frac{\|\bm{Q}_F(\bm{y}-\bar{y}\bm{1}_n)\|^2}{\|\bm{y}-\bar{y}\bm{1}_n\|^2},
\end{equation*}
with $ \bm{Q}_\bgm=\bm{I}-\bm{X}_\bgm(\bm{X}'_\bgm\bm{X}_\bgm)^{-1}\bm{X}'_\bgm$
and $ \bm{Q}_F=\bm{I}-\bm{X}_F(\bm{X}'_F\bm{X}_F)^{-1}\bm{X}'_F$,
which are the coefficient of determination 
under the submodel $\mcM_{\bgm}$ and the full model $\mcM_F$, respectively.
From \eqref{BF:intro}, 
the Bayes factor does not depend on the particular \sss sampling density.
Hence, even when there is no specific information about 
the form of the error distribution of each model 
(other than spherical symmetry), 
it is not necessary to specify the exact form of the sampling density. 
It suffices to assume it is Gaussian.
As far as we know, in the area of Bayesian variable selection with 
shrinkage priors or Zellner's $g$-priors,
the sampling density has been assumed to be Gaussian and
this kind of robustness result has not yet been studied.
Note that we use the term ``robustness''
in this sense of distributional robustness over the class of \sss error distributions.
We specifically are not using the term to indicate a high breakdown point.
The use of the term ``robustness'' in our sense is however common
(if somewhat misleading) in the context of insensitivity to the
error distribution in the shrinkage literature.
In Section \ref{sec:Laplace}, by use of the Laplace approximation,
we approximate the Bayes factor given by \eqref{BF:intro} as 
$ \BF^G\gvf(\nu)\approx\widetilde{\BF}^G\gvf(\nu)$
\begin{equation}\label{eq:intro.approxBF}
\widetilde{\BF}^G\gvf(\nu) =
\frac{\varphi(q_{\bgm}-\nu,1-R^2_{\bgm})}{\varphi(p-\nu,1-R^2_F)}
\BF^G\gvf\BIC
\end{equation}
as $n\to\infty$ where 
\begin{equation*}
\varphi(s,r)= rs^{s-1}\{(1/r-1)e\}^{-s}
\end{equation*}
and $\BF^G\gvf\BIC$ is the BIC based alternative 
for standard Gaussian errors
\begin{equation*}
\BF^G\gvf\BIC= \left\{\frac{(1-R_{\bgm}^2)^{-n}n^{-q_{\bgm}}}{(1-R_F^2)^{-n}n^{-p}}\right\}^{1/2}.
\end{equation*}
(See, e.g.~\cite{Hastie-etal-2009}, Chapter 7.) 
Since $\varphi_0(s,r)$ does not depend on $n$, 
\eqref{eq:intro.approxBF} is asymptotically equivalent to BIC with
a simple $O(1)$ rational correction function
depending upon $\nu$ as well as the $R$-squares and the numbers of predictors.
Actually this is a special case of Theorem \ref{thm:Laplace} in which
several Bayes factors under Gaussian errors which have been 
proposed in earlier studies, are shown to have similar asymptotic
approximations.
While the main theme in this paper is to develop the relationship 
\eqref{BF:intro} under sub-harmonic priors, we believe 
that this asymptotic equivalence is another noteworthy contribution, in particular, from 
the computational point of view.
In Section \ref{sec:consistency}, we show that our Bayes factor 
has model selection consistency 
uniformly over the class of \sss error distributions,
as $n \to \infty$ and $p$ is fixed.
It also follows from these results that several model selection 
methods recently studied in the literature for Gaussian errors
have model selection consistency for the entire class of \sss error distributions.
In Section \ref{sec:null}, we indicate an alternative development
which allows all $2^p$ possible models.
We emphasize once more that the inclusion of all $2^p$ possible models
(including the null model) requires only a slight modification
of the developments in the main body of the paper and that
resulting Bayes factor is of the form given in \eqref{eq:intro.approxBF}
with $n-1$ replaced by $n$, $q_{\bgm}$ by $q_{\bgm}+1$, $p$ by $p+1$, 
$R^2_{\bgm}$ by $\check{R}_{\bgm}^2$, and $R^2_F$ by $\check{R}_F^2$
where 
\begin{equation*}
 \check{R}_{\bgm}^2
=1-\frac{\|\bm{Q}_\bgm(\bm{y}-\bar{y}\bm{1}_n)\|^2}{\|\bm{y}\|^2},\quad
 \check{R}_{F}^2
=1-\frac{\|\bm{Q}_F(\bm{y}-\bar{y}\bm{1}_n)\|^2}{\|\bm{y}\|^2},
\end{equation*}
are the ``non-centered'' coefficients of determination.
Further, model consistency (including for the null model) holds for this modification.
We provide illustrations of the method and comparisons with other methods
using both simulated and real data in Section \ref{sec:ex}.
We give concluding remarks in Section \ref{sec:cr}. 
The Appendix presents some of the more technical proofs.

\section{Prior distributions}
\label{sec:prior}
In this section, for each submodel, we give a prior joint density of a form
\begin{equation*}
\pi(\alpha,\bbe_{\bgm}, \eta_{\bgm})= \pi(\alpha)\pi(\eta_{\bgm})
\pi(\bbe_{\bgm}|\eta_{\bgm}) ,
\end{equation*}
where $\eta_{\bgm}=1/\sigma_{\bgm}^2$.
We choose the natural priors for location ($\alpha$) and scale ($\eta_{\bgm}$), 
\begin{equation} \label{improper-prior-alpha}
\pial(\alpha)=I_{(-\infty,\infty)}(\alpha),
\end{equation}
and 
\begin{equation} \label{improper-prior-sigma^2}
\piet(\eta_{\bgm})=\eta_{\bgm}^{-1}I_{(0,\infty)}(\eta_{\bgm}).
\end{equation}
Since \eqref{improper-prior-alpha} and \eqref{improper-prior-sigma^2} 
have invariance to location and scale transformation,
respectively, they are considered by many as non-informative objective priors. 

Next we give conditional priors on
$\bbe_{\bgm}$ given $\eta_{\bgm}$
\begin{equation}\label{g-prior-1}
\pibeet(\bbe_{\bgm}|\eta_{\bgm};\nu) 
= \int_0^\infty \pig(g;\nu) 
\phi_{q_{\bgm}}(\bbe_{\bgm}| \bm{0},
g\eta_{\bgm}^{-1}(\bm{X}'_{\bgm}\bm{X}_{\bgm})^{-1})dg
\end{equation}
where $\nu$ is a non-random positive parameter. Further
\begin{equation}\label{g-prior-2}
 \pig(g;\nu) = g^{\nu/2-1}I_{(0,\infty)}(g),
\end{equation}
and $\phi_q(\cdot| \bm{\mu}, \bm{\Sigma})$
denotes the $q$-variate Gaussian density with mean vector $\bm{\mu}$ and
covariance matrix $\bm{\Sigma}$.
The prior \eqref{g-prior-1} clearly has a hierarchical structure and 
it can be interpreted a scale mixture of Zellner's $g$-priors.
Similar priors have been considered by \cite{Liang-etal-2008} and
\cite{Maruyama-George-2011} and others under the Gaussian linear regression setup.
See Sub-Section \ref{sec:g-prior} below for a review of priors on $g$.
Thanks to the simple form of $ \pig(g|\nu)$, the analytical integration
is possible as
\begin{equation}\label{eq:limit.variant}
\begin{split}
& \pibeet(\bbe_{\bgm}|\eta_{\bgm};\nu) \\ 
& = \int_0^{\infty} g^{\nu/2-1} 
\frac{|\bm{X}'_{\bgm}\bm{X}_{\bgm}|^{1/2}\eta_{\bgm}^{q_{\bgm}/2}}{(2\pi)^{q_{\bgm}/2}g^{q_{\bgm}/2}}
\exp\left(-\frac{\eta_{\bgm}}{2g}\bbe'_{\bgm}\bm{X}'_{\bgm}\bm{X}_{\bgm}\bbe_{\bgm}\right)dg \\
&= \frac{\Gamma(\{q_{\bgm}-\nu\}/2)}{2^{\nu/2}\pi^{q_{\bgm}/2}}|\bm{X}'_{\bgm}\bm{X}_{\bgm}|^{1/2} 
(\bbe'_{\bgm}\bm{X}'_{\bgm}\bm{X}_{\bgm}\bbe_{\bgm})^{-(q_{\bgm}-\nu)/2}\eta_{\bgm}^{\nu/2},
\end{split}
\end{equation}
when $0<\nu<q_{\bgm}$, which we assume throughout the paper.
In summary, 
the prior joint density under $\mcM_{\bgm}$ is given by
\begin{equation} \label{improper-joint}
\begin{split}
& \pigm(\alpha,\bbe_{\bgm},\eta_{\bgm};\nu) 
 =\pial(\alpha)\piet(\eta_{\bgm})
\pibeet(\bbe_{\bgm}|\eta_{\bgm};\nu) \\
& = 
\frac{\Gamma(\{q_{\bgm}-\nu\}/2)}{2^{\nu/2}\pi^{q_{\bgm}/2}}|\bm{X}'_{\bgm}\bm{X}_{\bgm}|^{1/2} 
 (\bbe'_{\bgm}\bm{X}'_{\bgm}\bm{X}_{\bgm}\bbe_{\bgm})^{-(q_{\bgm}-\nu)/2}\eta_{\bgm}^{\nu/2-1},
\end{split}
\end{equation}
where $0<\nu<q_{\bgm}$.
\begin{remark}
The use of Bayes factors for model comparisons with these improper priors 
for $\alpha$ and $\eta$ is formally justified here 
because $\alpha$ and $\eta$ are location-scale parameters that appear in every submodel.   
See \cite{Ber-Per-Var-1998, Berger-Bernardo-Sun-2009} for details.  
Further, impropriety of $ \pibeet(\bbe_{\bgm}|\eta_{\bgm};\nu)$ comes from
that of $\pig(g;\nu)$. Clearly $\pig(g;\nu)$ appears in every submodel and
with the same $\nu$. Thus the use of the improper prior is justifiable.
\end{remark}
If, in the above joint prior on $(\alpha,\bbe_{\bgm},\eta_{\bgm})$, we make 
the change of variables, $\bm{\theta}_{\bgm}=(\bm{X}'_{\bgm}\bm{X}_{\bgm})^{1/2}\bbe_{\bgm}$,
the joint prior of $(\alpha,\bm{\theta}_{\bgm},\eta_{\bgm})$ becomes
\begin{equation} \label{improper-joint-1}
 \pigm(\alpha,\bm{\theta}_{\bgm},\eta_{\bgm};\nu) 
 = 
\frac{\Gamma(\{q_{\bgm}-\nu\}/2)}{2^{\nu/2}\pi^{q_{\bgm}/2}}
\|\bm{\theta}_{\bgm}\|^{-(q_{\bgm}-\nu)}\eta_{\bgm}^{\nu/2-1}.
\end{equation}
As noted in Section \ref{sec:intro},
the part depending on $\bm{\theta}_{\bgm}$, $\|\bm{\theta}_{\bgm}\|^{-(q_{\bgm}-\nu)}$ for $0 <\nu<\min(2,q_{\bgm})$,
is known as a sub-harmonic function, 
that is,
\begin{equation*}
 \sum_{i=1}^{q_{\bgm}} \frac{\partial^2}{\partial \theta_i^2}
\|\bm{\theta}_{\bgm}\|^{-(q_{\bgm}-\nu)} 
 = (q_{\bgm}-\nu)(2-\nu)\|\bm{\theta}_{\bgm}\|^{-(q_{\bgm}-\nu)-2}>0.
\end{equation*}

\subsection{Review of priors on $g$} \label{sec:g-prior}
As noted above, the prior given by \eqref{g-prior-1}
is a scale mixture of Zellner's $g$-priors.
Actually the original Zellner's $g$-priors were used
for the Gaussian linear regression setup and historically 
the hyperparameter $g$ has been a priori fixed or somehow estimated.
The first paper to effectively use a prior on $g$
was \cite{Zellner-Siow-1980}; 
they stated things in terms of multivariate Cauchy densities, 
which can always be expressed as a mixture of $g$-priors where
\begin{equation*}
 \pi(g)= g^{-3/2}\exp(-n/\{2g\}).
\end{equation*}
Here we review the prior on $g$, the second stage of $g$-priors 
for Gaussian linear regression. 
We hope that it helps to clarify our prior \eqref{g-prior-2} on $g$ which applies to the entire class of \sss error densities and not just Gaussian errors.

As a generalization of $\pig(g;\nu)$ given by \eqref{g-prior-2},
consider the prior on $g$,
\begin{equation} \label{nu_a}
\pig(g;\{\nu,k,l\})=
g^{\nu/2-1}(1+l/g)^{-k/(2l)}I_{(0,\infty)}(g).
\end{equation}
Note that it is improper at $g=0$ when $ \nu\leq -k/l$.
As we will see in Section \ref{sec:marginal} and \ref{sec:Laplace},
Bayes factors are not well-defined when the prior on $g$ is improper
at $g=0$ and that is why $ \nu > -k/l$ is assumed.
On the other hand, it is improper at $g=\infty$ when $\nu\geq 0$.
But  as long as $0\leq\nu< q_{\bgm}$,
the Bayes factor under $\mcM_{\bgm}$ is well-defined as shown in Section \ref{sec:marginal}. 

All examples of mixtures of $g$-priors that we have found in the literature
are of the form \eqref{nu_a}.
For example, 
\cite{Liang-etal-2008} considered the (proper) ``hyper-$g$'' case where $-2<\nu<0$,
$k=2-\nu$ and $l=1$ and the (proper) ``hyper-$g/n$'' case where $-2<\nu<0$,
$k=n(2-\nu) $ and $l=n$.
\cite{Guo-Speckman-2009} and \cite{Celeux-Anbari-Marin-Robert-2012}
considered the (improper) case where $\nu=0$, $k=2$ and $l=1$.
\cite{Zellner-Siow-1980} considered the case where $\nu=-1$, $ k=n$ and $l=0$ since
\begin{equation*}
\lim_{l\to 0}(1+l/g)^{-k/(2l)} =\exp(-\{k/(2g)\}).
\end{equation*}
\cite{Maruyama-George-2011} considered $-1<\nu<0$, $k=(n-q_{\bgm}-1)/2$ and $l=1$. 
As in \eqref{g-prior-2}, we are considering the (improper) 
case where $0<\nu<q_{\bgm}$ and $k=0$.
\begin{remark}\label{rem:separability2}
As emphasized in  \cite{Liang-etal-2008}, a major reason for studying
the above class of mixtures of $g$-priors is their tractability in the
Gaussian case. One of the main results of this paper is that
under a condition of separability (see e.g.~Theorem \ref{thm:main-BF})
the Bayes factors for the normal case are valid for
the entire class of \sss distributions. We show also, in Lemma \ref{lem:uniqueness} below, that
our class of sub-harmonic priors are the only ones 
among the class of mixtures of $g$-priors given in \eqref{nu_a}.
for which this extension holds.
In this sense, our sub-harmonic priors are unique among the class \eqref{nu_a},
in being robust over the class of \sss distributions and leading to
tractable Bayes factors.
\end{remark}

\begin{remark}\label{rem:objectivity}
Here we discuss objectivity (or at least non-subjectivity)
of the prior in terms of hyper-parameters of the prior on $g$.
Under the prior for $g$ given by \eqref{nu_a}, consider the (proper or improper) prior on
$\bbe_{\bgm}$ given $\eta_{\bgm}$, 
\begin{equation}\label{bbe|eta}
\begin{split}
\pibeet(\bbe_{\bgm};\eta_{\bgm},\{\nu,k,l\}) 
\stackrel{\mathrm{def.}}{=}
\int_0^\infty \pi(g;\{\nu,k,l\})\phi_{q_{\bgm}}(\bbe_{\bgm}| \bm{0},
g\eta_{\bgm}^{-1}(\bm{X}'_{\bgm}\bm{X}_{\bgm})^{-1})dg
\end{split}
\end{equation}
In order to obtain the asymptotic behavior of the density as
$ \eta_{\bgm}^{1/2}\|\bbe_{\bgm}\|\to\infty$,
we appeal to the Tauberian theorem for the Laplace transform 
(see \cite{geluk-dehaan-1987}).
Since $ (1+l/g)^{k/(2l)}\to 1$ for any $\{k,l\}$ as $g\to\infty$,
\begin{equation} \label{eq:asymp.beta}
 \lim_{\eta_{\bgm}^{1/2}\|\bbe_{\bgm}\|\to\infty}
\frac{\pibeet(\bbe_{\bgm}|\eta_{\bgm},\{\nu,a,b\})}
{\left\{\eta_{\bgm}\bbe'_{\bgm}\bm{X}'_{\bgm}
\bm{X}_{\bgm}\bbe_{\bgm}\right\}^{-\frac{q_{\bgm}-\nu}{2}}
\eta_{\bgm}^{\frac{q_{\bgm}}{2}}} 
=
\frac{\Gamma(\{q_{\bgm}-\nu\}/2)}{2^{\nu/2}\pi^{q_{\bgm}/2}}|\bm{X}'_{\bgm}\bm{X}_{\bgm}|^{1/2}
\end{equation}
when $\nu<q_{\bgm}$.
Hence the asymptotic order of \eqref{eq:asymp.beta} is 
the same as \eqref{eq:limit.variant} and does not depend on $\{k,l\}$.
The larger $\nu(<q_{\bgm})$ is, the more objective the prior
$ \pibeet(\bbe_{\bgm}|\eta_{\bgm},\{\nu,k,l\})$ is.
\end{remark}

\section{Marginal density and Bayes factor under sub-harmonic priors}
\label{sec:marginal}
In this section we derive the marginal density under each submodel
and the Bayes factor for comparing
each $ \mcM_{\bgm}$ to the full model $\mcM_F$.
The marginal density of $\bm{y}$ under
$\mcM_{\bgm} $,
is given by
\begin{equation}
\begin{split}
M_{\bgm}(\bm{y}|\nu)= & \int_{-\infty}^{\infty} \int_{R^{q_{\bgm}}}
 \int_{0}^{\infty}
\eta_{\bgm}^{n/2}f_n(\eta_{\bgm}\|\bm{y}-\alpha \bm{1}_n -\bm{X}_{\bgm}\bbe_{\bgm}\|^2) \\
& \qquad \times \pigm(\alpha, \bbe_{\bgm}, \eta_{\bgm};\nu)
d \alpha \, d \bbe_{\bgm} \, d\eta_{\bgm} , 
\end{split}
\label{full-marginal}
\end{equation}
where the prior $\pigm(\alpha,\bbe_{\bgm},\eta_{\bgm}|\nu)$ is
given by \eqref{improper-joint}:
\begin{equation} 
 \pigm(\alpha,\bbe_{\bgm},\eta_{\bgm};\nu) 
 = 
\frac{\Gamma(\{q_{\bgm}-\nu\}/2)}{2^{\nu/2}\pi^{q_{\bgm}/2}}
\frac{|\bm{X}'_{\bgm}\bm{X}_{\bgm}|^{1/2}\eta_{\bgm}^{\nu/2-1}} 
{ (\bbe'_{\bgm}\bm{X}'_{\bgm}\bm{X}_{\bgm}\bbe_{\bgm})^{(q_{\bgm}-\nu)/2}}.
\end{equation}
Two aspects of the joint density $\pigm(\alpha,\bbe_{\bgm},\eta_{\bgm};\nu)$ above,
namely
\begin{enumerate}[{K}1.]
\item
$(\alpha,\bbe_{\bgm})$ and $\eta_{\bgm}$ are separable
in the sense that the prior distribution factors into two terms involving the indicated parameters, \label{Key1}
\item 
The term involving $\eta_{\bgm}$ is given by a power function, \label{Key2}
\end{enumerate}
will be key for calculating the marginal density for 
the entire class of \sss error densities and not just Gaussian errors.
Let $ M_{\bgm}(\bm{y}|\nu) $ and $ M_{\bgm}^G(\bm{y}|\nu) $ be the marginal densities 
under $\mcM_{\bgm}$ with general \sss errors $\bep_\bgm$ and with
with standard Gaussian errors $\bep_G$, respectively.
The next result provides a relationship between $ M_{\bgm}(\bm{y}|\nu) $ and $ M^G_{\bgm}(\bm{y}|\nu) $.
\begin{lemma}\label{lem:relationship-G}
Let $\nu$ be between $0$ and $q_{\bgm}$. 
Assume the existence of $E[\|\bep_{\bgm} \|^\nu]$. Then
\begin{equation} \label{MG}
M_{\bgm}(\bm{y}|\nu)
=\frac{E[\|\bep_{\bgm} \|^\nu]}{E[\|\bep_G \|^\nu]}
M_{\bgm}^G(\bm{y}|\nu).
\end{equation}
\end{lemma}
\begin{proof}
 See Appendix \ref{app:relationship-G}.
\end{proof}
Hence $ M_{\bgm}(\bm{y}|\nu)$ depends on the error distribution $ \bep_{\bgm}$
only through the $\nu$-th moment of $\bep_{\bgm}$, $ E[\|\bep_{\bgm} \|^\nu]$. 
The identity in Lemma \ref{lem:relationship-G} under K\ref{Key1} and K\ref{Key2} has been used in \cite{Maruyama-2003b}  and \cite{Maru-Straw-2005} 
for finding robust minimax estimators. 
Essentially the same identity in Lemma \ref{lem:relationship-G} 
under scale mixture of normals error distribution $\bep_{\bgm}$ with the prior 
$\pi(\alpha,\bbe_{\bgm},\eta_{\bgm})=\eta^{-1}_{\bgm}$ was used 
for the Bayesian prediction problem in \cite{Jammalamadaka-et-al-1987}.

Using the expression of the prior \eqref{improper-joint} as
the scale mixture of normals in Section \ref{sec:prior},
we will make use of the following result which may be founded in equation (5) of
\cite{Liang-etal-2008}.
\begin{lemma} \label{lem:m-G}
Let $0<\nu<q_{\bgm}$. Then
\begin{equation}\label{MGG}
M^G_{\bgm}(\bm{y}|\nu)
 = \frac{n^{1/2}\Gamma(\{n-1\}/2)}{\|\bm{y}-\bar{y}\bm{1}_n\|^{n-1}\pi^{(n-1)/2}}
\int_0^{\infty} \frac{g^{\nu/2-1}(1+g)^{(n-q_{\bgm}-1)/2}}
{\left\{g(1-R_{\bgm}^2)+1 \right\}^{(n-1)/2}} \, dg, 
\end{equation}
where $R_{\bgm}^2$ is the coefficient of determination under the submodel $\mcM_{\bgm}$.
\end{lemma}
Combining Lemmas \ref{lem:relationship-G} and \ref{lem:m-G}, 
we have the main result of this paper.
\begin{thm}\label{thm:main-BF}
Assume the full model $\mcM_F$ and the submodel $\mcM_\bgm$ are
given by \eqref{full-model} and \eqref{submodel-gamma}, respectively.
Also assume their error terms, $\bep_F$ and $\bep_\bgm$ have the same
\sss distribution \eqref{bep_sim_f} with mean zero and the identity
covariance matrix.
Let $0<\nu <q_{\bgm}$.
Assume that the proper joint prior densities
of $(\alpha, \bbe_F, \eta_F)$ and $(\alpha, \bbe_{\bgm}, \eta_{\bgm})$
are given by \eqref{improper-joint} and assume also $E[\|\bep_{\bgm}\|^\nu]<\infty$.
Then, for $\mcM_\bgm\neq\mcM_N$, the Bayes factor for comparing
each of $ \mcM_{\bgm}$ to the full model $ \mcM_F$ is
given by
\begin{equation}\label{eq:main-thm-BF}
\BF\gvf(\nu)
 = \frac{M_{\bgm}(\bm{y}|\nu)}{M_F(\bm{y}|\nu)} 
=\BF^G\gvf(\nu)
\end{equation}
where
\begin{equation} 
\label{main-BF}
\BF^G\gvf(\nu)
=\frac{\int_0^{\infty} g^{\frac{\nu}{2}-1}(1+g)^{\frac{n-q_{\bgm}-1}{2}}
\{g(1-R_{\bgm}^2)+1 \}^{-\frac{n-1}{2}}\, dg}
{\int_0^{\infty} g^{\frac{\nu}{2}-1}(1+g)^{\frac{n-p-1}{2}}
\{g(1-R_F^2)+1 \}^{-\frac{n-1}{2}}\, dg}.
\end{equation}
\end{thm}
By Theorem \ref{thm:main-BF}, even when there is no specific information about 
the error distribution of each model (other than spherical symmetry), 
but we assume they are all the same, 
it is not necessary to specify the exact form of the sampling density. 
It suffices to assume they are all Gaussian.
As far as we know, in the area of Bayesian variable selection with shrinkage priors,
the sampling density has been assumed to be Gaussian and
this kind of robustness result has not yet been studied.
\begin{remark}
\cite{Maruyama-George-2011} considered Bayesian variable selection under Gaussian errors.
They proposed Bayes factors with a simple analytic form
under generalized ridge-type priors.
The results heavily depend on special features of Gaussian distributions and hence
the extension or generalization of \cite{Maruyama-George-2011}
to the general \sss case, may not be possible, 
or may not lead to analytically tractable procedures
which are distributionally robust to \sss error distributions.
\end{remark}

\begin{remark}\label{rem:coherent}
A collection of Bayes factors is called coherent if 
\begin{equation*}
 \BF_{\gamma_1:\gamma_2}=\BF_{\gamma_1:\gamma_0}\BF_{\gamma_0:\gamma_2}, 
\mbox{ and }\BF_{\gamma_1:\gamma_2}=1/\BF_{\gamma_2:\gamma_1},
\end{equation*}
for all $\gamma_1$ and $\gamma_2$ (see, e.g.~\cite{Robert-2007-BC}).
By \eqref{eq:main-thm-BF}, the Bayes factors corresponding
to our sub-harmonic priors are coherent 
(with the exception of those involving the null model $\mcM_N$),
which is why we require $\PR(\mcM_N)=0$.
Also with the adaptation of the alternative specification given in
Section \ref{sec:null}, coherence holds for all Bayes factors including
those involving the null model $\mcM_N$.

As in \cite{Zellner-Siow-1980} and \cite{Liang-etal-2008}, 
the posterior probability of any model is an expression of the 
form \eqref{posterior-2} given by
\begin{equation*}
 \mathrm{Pr}(\mcM_{\bgm}|\bm{y})
=\frac{\PR(\mcM_{\bgm}) \BF\gvf}{\sum_{\bgm} \PR(\mcM_{\bgm}) \BF\gvf}.
\end{equation*}
In our development, we choose the full model $\mcM_F$ as the base model
rather than the null model $\mcM_N$ employing the encompassing approach of
\cite{Zellner-Siow-1980}. Without employing the adaptation of our prior
described in Section \ref{sec:null},
the choice of the full model $\mcM_F$ as the base model, 
as opposed to the null model $\mcM_N$, is forced on us since $\PR(\mcM_N)=0$. 
\cite{Liang-etal-2008} argue that the null model is the 
superior choice as a base model in their setup (which also involves $g$-priors
or mixture thereof) due to incoherence of the Bayes factors 
if $\mcM_F$ is chosen as the base model (in their setup).
This incoherence arises in the setup of \cite{Liang-etal-2008} because
prior distribution on the full model $\mcM_F$ depends on each nested
alternative $\mcM_{\bgm}$. This incoherence is not a problem in our setup
since the choice of prior for each submodel depends only on the submodel,
and we have taken care that all relevant (conditional) posteriors are
well defined. 
As noted above, by \eqref{eq:main-thm-BF} our Bayes factors are coherent.
In fact our development (aside from eliminating the null model from the 
consideration) is very close in spirit to the null-based Bayes factors 
approach in \cite{Liang-etal-2008}. 
\end{remark}

\subsection{Robustness, uniqueness and tractability}
For establishing Lemma \ref{lem:relationship-G},
two aspects K\ref{Key1} (separability) and K\ref{Key2} 
(power function for the distribution of $\eta_{\bgm}$)
are key for calculating the marginal density for 
the entire class of \sss error densities and not just Gaussian errors.
Recall from Sub-Section \ref{sec:g-prior} that
the class of mixtures of $g$-priors found in the literature are of the form
\eqref{g-prior-1} with $\pi(g;\cdot)$ of the form \eqref{nu_a}.
Hence in each case we may express $ \pibeet(\bbe_{\bgm}|\eta_{\bgm})$ as 
\begin{equation}\label{eq:mixturesg}
\pibeet(\bbe_{\bgm}|\eta_{\bgm}) 
= \int_0^\infty g^{\nu/2-1}\tilde{\pi}(g) 
\phi_{q_{\bgm}}(\bbe_{\bgm}| \bm{0},
g\eta_{\bgm}^{-1}(\bm{X}'_{\bgm}\bm{X}_{\bgm})^{-1})dg
\end{equation}
where $\tilde{\pi}(g)\to 1$ as $g\to\infty$.
Note in particular that our sub-harmonic prior (see \eqref{g-prior-2})
corresponds to $\tilde{\pi}(g)\equiv 1$ in \eqref{eq:mixturesg}.
If $ \pibeet(\bbe_{\bgm}|\eta_{\bgm}) $ is separable with respect to
$\bbe_{\bgm}$ and $\eta_{\bgm}$ and
the term involving $\eta_{\bgm}$ is given by the power function,
K\ref{Key1} and K\ref{Key2} are satisfied.

Clearly when $ \tilde{\pi}(g)\equiv 1$, K\ref{Key1} and K\ref{Key2} are satisfied
as in \eqref{eq:limit.variant}. 
As in \eqref{eq:asymp.beta} in Sub-Section \ref{sec:g-prior},
the asymptotic order of $\pibeet(\bbe_{\bgm}|\eta_{\bgm}) $ with 
$\lim_{g\to\infty}\tilde{\pi}(g)=1$ is 
the same as $\pibeet(\bbe_{\bgm}|\eta_{\bgm}) $ with $ \tilde{\pi}(g)\equiv 1$.
As a matter of fact, among the class of scale mixtures of $g$-priors 
\eqref{eq:mixturesg} with $ \lim_{g\to\infty}\tilde{\pi}(g)=1$,
$\tilde{\pi}(g)\equiv 1$ is the only choice leading to separability
as stated in the following result.
In particular, our sub-harmonic priors are the only separable priors
among the class of mixtures of $g$-priors found in the literature and
more generally among priors of the form \eqref{eq:mixturesg} where 
$\lim_{g\to\infty}\tilde{\pi}(g)=1$.
\begin{lemma}\label{lem:uniqueness}
 If $\pi(\bbe_{\bgm}|\eta_{\bgm})$ in \eqref{eq:mixturesg}
is continuous and can be expressed as the product of a function
of $\bbe'\bm{X}'\bm{X}\bbe$ and a function of $\eta$ and if
$\lim_{g\to\infty}\tilde{\pi}(g)=1$, then $\tilde{\pi}(g)\equiv 1$.
\end{lemma}
\begin{proof}
See Appendix \ref{app:uniqueness}.
\end{proof}
The relationship given in Lemma \ref{lem:relationship-G} 
remains true under more general separable priors
$\eta_{\bgm}^{\nu/2-1}\pi(\alpha, \bbe_{\bgm})$, which satisfy
K\ref{Key1} and K\ref{Key2}. Unfortunately we could not find any other priors, 
which lead to analytically tractable Bayes factors under Gaussian errors,
except for our sub-harmonic priors $ \pigm(\alpha,\bbe_{\bgm},\eta_{\bgm};\nu) $
given by \eqref{improper-joint}.
For simplicity, let 
\begin{equation*}
 \pi(\alpha, \bbe_{\bgm})=\pi(\alpha)\pi(\bbe)=1\times |\bm{X}'\bm{X}|^{1/2}\pi(\bbe'\bm{X}'\bm{X}\bbe).
\end{equation*}
After some calculation, the marginal density under the prior is given by
\begin{equation*}
 \frac{\Gamma((n+\nu-1)/2)}{2^{-\nu/2}\pi^{(n-1)/2}}
\iint_{R^{q_{\bgm}}}\frac{\pi(\|\bm{\theta}_{\bgm}\|^2)d\bm{\theta}_{\bgm}}
{\|\bm{\theta}_{\bgm}-\bm{U}'_{\bgm}(\bm{y}-\bar{y}\bm{1}_n)\|^{n+\nu-1}}
\end{equation*}
which is a function of $R^2_{\bgm}=\|\bm{U}'_{\bgm}(\bm{y}-\bar{y}\bm{1}_n)\|^2$
from spherical symmetry, (here $\bm{U}_{\bgm}$ has orthonormal columns which
span the column space of $\bm{X}_{\bgm}$).
In other words, the marginal density under a general separable prior
is expressed as a multiple integral, which is less tractable
than the marginal density under sub-harmonic separable prior with a one-dimensional
integral expression as in Lemma \ref{lem:m-G} and Theorem \ref{thm:main-BF}.

In summary, as far as we know, the sub-harmonic prior is the unique choice for
robustness and tractability. Certainly this is so among the class of
mixtures of $g$-priors in Sub-Section \ref{sec:g-prior} which represent the class studied
in the literature for Gaussian samples.

\subsection{The choice of $\nu$}
 In an earlier version of this paper we developed the results in the more
general context wherein the \sss distribution of 
$\bep_\bgm \sim f_\bgm(\|\bep_\bgm\|^2)$ could depend on $\bm{\gamma}$, i.e.~it
could be different for each submodel $\mcM_\bgm$. 
So $\bep_F\sim f_F(\|\bep_F\|^2)$ as well.
All of the above results can be developed for the more general case. 
The only essential changes
are that \eqref{eq:main-thm-BF} in Theorem \ref{thm:main-BF} becomes
\begin{equation}\label{eq:main-thm-BF-1}
\BF\gvf(\nu)
=\frac{E[\|\bep_{\bgm} \|^\nu]}{E[\|\bep_F \|^\nu]}
\BF^G\gvf(\nu).
\end{equation}
We investigated the ranges of these ``correction terms'' in 
\eqref{eq:main-thm-BF-1} 
when $\bep_\bgm$ and $\bep_F$ have possibly different \sss $t$-distributions
with at least $3$-degrees of freedom (so that the variances exist),
both analytically and numerically. We found that 
$\BF\gvf(\nu)$ was independent of $n$ and reasonably stable for all $\nu$
in range $(0,1)$ but that stability was greater for $\nu$ close to $0$.
This trade-off between stability (favoring $\nu\approx 0$) and 
objectivity (favoring larger $\nu$) led us initially 
to prefer the midpoint of the allowable values in $(0,1)$, 
namely $ \nu=1/2$ as the default choice.
However the examples presented in Section \ref{sec:ex} indicate that
the performance of the method seems insensitive to the choice of $\nu$
in the range of $(0,1)$.

It is interesting to note in connection with the above that
the correction term $E[\|\bep_{\bgm} \|^\nu]/E[\|\bep_F \|^\nu]$
for $\BF\gvf(\nu)$ approaches $1$ as $\nu\to 0$.
Hence choices of $\nu$ close to $0$ are essentially completely robust
to choice of \sss error distribution for the submodels.

Note also that if we force $\PR(\mcM_\bgm)=0$ for all submodels such that
$q_\bgm\leq 2$, then the allowable range of $\nu$ is $(0,3)$ and
hence $\nu=2$ becomes a possible choice. In this case, again,
the correction term $E[\|\bep_{\bgm} \|^\nu]/E[\|\bep_F \|^\nu]=1$
regardless of the choice of error distributions, since
we have assumed the variance of each component $\bep_{\bgm}$ is $1$, and
the Bayes factor is completely robust to choice of \sss
error distribution. Additionally, the case $\nu=2$
corresponds to the harmonic prior
\begin{equation*}
\pigm(\alpha,\bbe_{\bgm},\eta_{\bgm};\nu) \propto 
\|\bm{\theta}_{\bgm}\|^{2-q_{\bgm}}
\end{equation*}
where $\bm{\theta}_{\bgm}=(\bm{X}'_{\bgm}\bm{X}_{\bgm})^{1/2}\bbe_{\bgm}$
and $q_{\bgm}\geq 3$.
It is well-known that the harmonic prior plays an important role
in estimation problems with the Stein effect
in the sense that the GB estimators of $\bbe_\bgm$ for such submodels are minimax.
See \cite{Maruyama-2003b} for details.
It is interesting to observe the additional advantage of 
the harmonic prior in the model choice problem.
See Section \ref{sec:cr} for some additional discussion of the advantages of such priors.

\subsection{BIC under spherically symmetric error distributions}
BIC (\cite{Schwarz-1978}) is a popular criterion for model selection. 
See e.g.~\cite{Hastie-etal-2009} Chapter 7. 
We will show in this subsection 
that BIC has a similar distributional robustness property to the
above Bayes model selection procedure.
In Section \ref{sec:Laplace}, we will develop Laplace approximations
to our Bayes factors which relate them to BIC.
In Section \ref{sec:consistency} we will show that both 
the BIC and our Bayes model selection 
(as well as a number of other Bayes methods developed specifically for the Gaussian case)
are consistent for the entire class of \sss models.

BIC for the model $\mcM_{\bgm}$ is defined as
\begin{equation} \label{BIC-1}
\BIC_\bgm=-2\ln
\left\{ 
\max_{\alpha,\bbe_{\bgm},\eta_{\bgm}}\eta_{\bgm}^{n/2}
f_n\left(
\eta_{\bgm}\|\bm{y}-\alpha \bm{1}_n -\bm{X}_{\bgm} \bbe_{\bgm} \|^2
\right)
n^{-q_{\bgm}/2}\right\},
\end{equation}
and is derived by eliminating $O(1)$ terms
from the approximate marginal densities.
Here we denote 
\begin{equation}
 M_{\bgm}(\bm{y}|\mathrm{BIC})=\exp(-\BIC_\bgm/2).
\end{equation}
In general, 
maximization with respect to unknown parameters in \eqref{BIC-1}
is not always tractable. However
when $\bep_{\bgm} $ has a unimodal 
\sss distribution,
the maximum is achieved at $\hat{\alpha}=\bar{y}$, 
$\hat{\bbe}_{\bgm}=(\bm{X}'_{\bgm}\bm{X}_{\bgm})^{-1}\bm{X}'_{\bgm}\bm{y}$, and 
\begin{equation}
 1/\hat{\eta}_{\bgm}=c\| \bm{y} -\hat{\alpha}\bm{1}_n -\bm{X}_{\bgm}\hat{\bbe}_{\bgm}\|^2
=c \|\bm{y} -\bar{y}\bm{1}_n\|^2(1-R_{\bgm}^2)
\end{equation}
where $c$ is the sole solution of 
\begin{equation} \label{c-gamma}
 n/2+cf'_n(c)/f_n(c)=0.
\end{equation}
Hence $ M_{\bgm}(\bm{y}|\mathrm{BIC})$ may be expressed as
\begin{equation}\label{MG-BIC}
M_{\bgm}(\bm{y}|\mathrm{BIC})
=\frac{c^{-n/2}f_n(c)}{c_G^{-n/2}f_G(c_G)}
M^G_{\bgm}(\bm{y}|\mathrm{BIC})
\end{equation}
where $ M^G_{\bgm}(\bm{y}|\mathrm{BIC})$ is $ M_{\bgm}(\bm{y}|\mathrm{BIC}) $ 
with Gaussian errors, specifically
\begin{equation} \label{MGG-BIC}
\begin{split}
M^G_{\bgm}(\bm{y}|\mathrm{BIC}) &=
c_G^{-n/2}f_G(c_G)\{ \|\bm{y} -\bar{y}\bm{1}_n\|^2(1-R_{\bgm}^2)\}^{-n/2}
n^{-q_{\bgm}/2} \\
&=  n^{-n/2}f_G(n)\{ \|\bm{y} -\bar{y}\bm{1}_n\|^2(1-R_{\bgm}^2)\}^{-n/2}
n^{-q_{\bgm}/2}
\end{split}
\end{equation}
(since $c_G$ is given by $n$).
Clearly \eqref{MG-BIC} and \eqref{MGG-BIC} correspond to \eqref{MG} and
\eqref{MGG}, respectively.
Hence we have the following result.
\begin{thm}\label{thm:BICrobust}
Assume the full model $\mcM_F$ and the submodel $\mcM_\bgm$ are
given by \eqref{full-model} and \eqref{submodel-gamma}, respectively.
Also assume their error terms, $\bep_F$ and $\bep_\bgm$ have a
unimodal \sss distribution \eqref{bep_sim_f} with the mean zero and the identity
covariance matrix.
Then the Bayes factor based on BIC 
for comparing each of $ \mcM_{\bgm}$ to the full model $ \mcM_F$ is given by
\begin{equation} \label{main-BIC}
\BF\gvf\BIC
=\frac{M_{\bgm}(\bm{y}|\mathrm{BIC})}{M_F(\bm{y}|\mathrm{BIC})}
=\BF^G\gvf\BIC
\end{equation}
where $\BF^G\gvf\BIC$ is the BIC based Bayes factor under 
Gaussian errors,
\begin{equation}\label{BF-G-BIC}
\BF^G\gvf\BIC
= \left\{\frac{(1-R_{\bgm}^2)^{-n}n^{-q_{\bgm}}}{(1-R_F^2)^{-n}n^{-p}}\right\}^{1/2}.
\end{equation}
\end{thm}
Obviously \eqref{BF-G-BIC} corresponds to \eqref{main-BF}.
By Theorem \ref{thm:BICrobust}, 
the Bayes factor based on BIC is also independent of 
the error distribution provided each distribution is unimodal and is
the same for all models (c.f.~Theorem \ref{thm:main-BF}).
Note that $\BF^G\gvf\BIC$ is well defined if $\mcM_\bgm=\mcM_N$.

\section{The Laplace approximation of BF under Gaussian errors}
\label{sec:Laplace}
In Section \ref{sec:marginal}, we saw that the Bayes factor $\BF\gvf(\nu) $
under \sss errors is equal to
$\BF^G\gvf(\nu) $, which is the Bayes factor under Gaussian errors.
In this section, we consider the so-called Laplace approximation
of some Bayes factors under Gaussian errors.
We will approximate 
not only the function $\BF^G\gvf(\nu) $ but also
Bayes factors with respect to more general priors where the prior on $g$ is
\eqref{nu_a};
\begin{equation}\label{prior-Laplace.sec}
\pi(g;\{\nu,k\}) = g^{\nu/2-1}(1+g)^{-k/2}.
\end{equation}
When the same prior on $g$ is used for $ \mcM_{\bgm}$ and $ \mcM_F$, 
improper choices of $\nu$ ($0\leq \nu<q_{\bgm}$) as well as proper
choices of $\nu$ 
($-k<\nu<0$) are valid for use.
Under Gaussian errors, the Bayes factor for comparing
each of $ \mcM_{\bgm}$ to $ \mcM_F$ 
is well-defined as
\begin{equation} 
\label{BF-1}
\begin{split}
&\BF^G\gvf[\nu,k] \\
&=\frac{\int_0^{\infty} g^{\nu/2-1}(1+g)^{-k/2}(1+g)^{\frac{n-q_{\bgm}-1}{2}}
\{g(1-R_{\bgm}^2)+1 \}^{-\frac{n-1}{2}}\, dg}
{\int_0^{\infty} g^{\nu/2-1}(1+g)^{-k/2}(1+g)^{\frac{n-p-1}{2}}
\{g(1-R_F^2)+1 \}^{-\frac{n-1}{2}}\, dg}
\end{split}
\end{equation}
where $k\geq 0$, $-k<\nu<q_{\bgm}$.

First we provide a summary of Laplace approximations to the integral 
based on \cite{Tierney-Kadane-1986}.
For integrals of the form 
\begin{equation*}
 \int_{-\infty}^\infty \exp(h(\tau,n))d\tau,
\end{equation*}
we make the use of the fully exponential Laplace approximation, based on
expanding a smooth unimodal function $h(\tau,n)$ in a Taylor series expansion
about $\hat{\tau}$, the mode of $h(\tau,n)$.
The Laplace approximation is given by 
\begin{equation}\label{Laplace1}
\lim_{n \to \infty} 
\frac
{\int_{-\infty}^\infty \exp(h(\tau,n))d\tau}
{(2\pi)^{1/2}\hat{\sigma}_h\exp(h(\hat{\tau},n))}
=1
\end{equation}
where
\begin{equation*}
 \hat{\sigma}_h=
\left\{-\frac{\partial^2 h(\tau,n)}{\partial\tau^2} \Big|_{\tau=\hat{\tau}} \right\}^{-1/2}.
\end{equation*}
In the following, we will use the symbol $f(n) \approx g(n) $ ($n \to \infty$)
if 
\begin{equation}
 \lim_{n \to \infty} \frac{f(n)}{g(n)}=1.
\end{equation}
Hence the approximation given by \eqref{Laplace1} is written as
\begin{equation}
\int_{-\infty}^\infty \exp(h(\tau,n))d\tau \approx 
(2\pi)^{1/2}\hat{\sigma}_h \exp(h(\hat{\tau},n)) , \ (n \to \infty).
\end{equation}

The next result gives approximations of the Bayes factor
\eqref{main-BF} in terms of the Bayes factor based on
BIC given in \eqref{BF-G-BIC}.
\begin{thm} \label{thm:Laplace}
Let the prior be given by \eqref{prior-Laplace.sec}. 
Assume $\{\nu,k\}$ do not depend on $n$.
\begin{enumerate}
 \item Assume $-k<\nu<q$, and $0<r<1$. Then 
\begin{equation}\label{eq:thmLaplace1}
\int_0^{\infty} \frac{g^{\nu/2-1}}{(1+g^{-1})^{k/2}}
\frac{(1+g)^{\frac{n-q-1}{2}}}{{(1+rg)^{\frac{n-1}{2}}}} dg \\
 \approx  
\left\{\frac{4\pi \varphi(q-\nu,r)}{n^{q-\nu}r^n}\right\}^{1/2},
\end{equation}
where $ \varphi(s,r)= rs^{s-1}\{(1/r-1)e\}^{-s}$.
\label{thm:Laplace:part1}
\item 
Assume that $-k<\nu<q_{\bgm}$,
that $R^2_F$ is strictly less than $1$
and that $\mcM_{\bgm}\neq\mcM_N$.
Then $ \BF^G\gvf[\nu,k] \approx \widetilde{\BF}^G\gvf(\nu)$ where
\begin{equation}
 \widetilde{\BF}^G\gvf(\nu)  =
\left\{
\frac{\varphi(q_{\bgm}-\nu,1-R_{\bgm}^2)}{\varphi(p-\nu,1-R_F^2)}\right\}^{1/2}\BF^G\gvf\BIC 
\end{equation} 
and $ \BF^G\gvf\BIC$ is the $\mathrm{BIC}$ based alternative under Gaussian errors.
\label{thm:Laplace:part2}
\end{enumerate}
\end{thm}
\begin{proof}
See Appendix \ref{app:Laplace}.
\end{proof}
Clearly the function $\varphi$ does not depend on $n$ and hence 
Theorem \ref{thm:Laplace} shows that $ \BF^G\gvf[\nu,k]$ 
is asymptotically equivalent to BIC with a simple $O(1)$ correction function
depending $\nu$ 
as well as $\{p,q_{\bgm}\}$ and the $R$-squares.
Although several fully Bayes factors for the variable selection problem
have been proposed in the literature, the relationship between the approximate Bayes
factors and naive BIC has not been shown to the authors' knowledge.
In this sense, while the main contributions in this paper 
are given in Section \ref{sec:marginal}, 
Theorem \ref{thm:Laplace} 
may be a practically useful contribution because of the simplicity of the 
approximate Bays factor.

In this section, we have considered general Bayes factors under 
Gaussian errors. Recall, by \eqref{BF-1}, that $\BF\gvf(\nu)$ the Bayes factor w.r.t.~sub-harmonic
priors, under \sss errors, is equal to $\BF^G\gvf[\nu,0]$
for $0<\nu<1$. 
Under Gaussian errors, 
\cite{Liang-etal-2008} recommended
the use of $ \BF^G\gvf[\nu,2-\nu]$ with $ -2<\nu<0$.
\cite{Guo-Speckman-2009} and \cite{Celeux-Anbari-Marin-Robert-2012}
recommended the use of $ \BF^G\gvf[0,2]$.
these Bayes factors may be approximated as follows.
\begin{corollary}\label{cor:approxBF}
\begin{equation*}
\begin{split}
&   \BF\gvf(\nu)=\BF^G\gvf[\nu,0] \approx \widetilde{\BF}^G\gvf(\nu) 
\mbox{ for  } 0<\nu<1. \quad \mbox{sub-harmonic prior}\\
&   \BF^G\gvf[0,2] \approx \widetilde{\BF}^G\gvf(0). \quad
\mbox{a version of \cite{Guo-Speckman-2009}}\\
&   \BF^G\gvf[\nu,2-\nu] \approx \widetilde{\BF}^G\gvf(\nu) 
\mbox{ for  } -2<\nu<0. \quad \mbox{\cite{Liang-etal-2008}} 
\end{split}
\end{equation*}
\end{corollary}
In Section \ref{sec:ex}, we will see how 
approximate Bayes factors $\widetilde{\BF}^G\gvf(\nu)$
work numerically and how sensitive they are to to the choice of $\nu$.

\section{Model selection consistency} \label{sec:consistency}
In this section, we consider model selection consistency in the case
where $p$ is fixed and as $n$ approaches infinity.
Let $\mcM_T$ be the true model,
\begin{equation*}
 \bm{y}=\alpha_T \bm{1}_n+\bm{X}_T\bbe_T+\sigma_T \bep.
\end{equation*}
Consistency for model choice is defined as
\begin{equation*}
\plim_{n \to \infty}
\mathrm{Pr}(\mcM_T |\bm{y})=1, 
\end{equation*}
where plim denotes convergence in probability and the probability
distribution is the sampling distribution under the true model
$\mcM_T$.
We will show that the Bayes factors considered in the previous sections
have model selection consistency under generally \sss errors.
The consistency property is clearly equivalent to
\begin{equation} \label{equiv-consistency}
\plim_{n \to \infty}
\BF_{\gamma:T}=
\plim_{n \to \infty}
\frac{\BF\gvf}
{\BF_{T:F}}
=0  \quad \forall \bgm \neq T.
\end{equation}
For model selection consistency, we make the following assumptions;
\begin{enumerate}[{A}1.]
\item  
$U_n=\|\bep\|^2/n$ is bounded 
in probability from below and from above, that is,
for any $c>0$ and any positive integer $n$, 
there exists an $M$ such that
\begin{equation*}
\textstyle{\mathrm{Pr}\left(M^{-1}< U_n < M\right) >1-c}.
\end{equation*} \label{AS1} 
\item 
The limit of the correlation matrix of $x_1,\dots,x_p$, 
$\lim_{n \to \infty}\bm{X}'_F\bm{X}_F/n$, exists and is positive definite. 
\label{AS2}
\end{enumerate}
A\ref{AS1} seems more general than necessary.
It appears that, by the law of large numbers,
$U_n$ ought to converge to $1$ in probability, but
this is not necessarily true if the error distribution is not Gaussian
since in that case the errors are not independent.
In the case of a scale mixture of Gaussians, $U_n$ approaches, in law,
a random variable $\xi$ which has the distribution of the mixing variable
of the variance. Even when the error distribution is not a scale
mixture of Gaussians, A\ref{AS1} appears to be a reasonable
and minimal assumption.
A\ref{AS2} is the standard assumption which also appears in \cite{Knight-Fu-2000}
and \cite{Zou-2006}. 
Under these mild assumptions, 
we have following preliminary results for proving the consistency.

\begin{lemma} \label{lem:prel}
Assume A\ref{AS1} and A\ref{AS2}. 
\begin{enumerate}
 \item \label{1:lem:prel}
Assume $\mcM_\bgm\neq\mcM_N$.
For any $0<k<1$ and any positive integer $n$, 
there exists a $c_1(\bgm, k)>2$ such that
\begin{equation}
 \mathrm{Pr}\left(\frac{1}{c_1(\bgm, k)} < R^2_{\bgm}
<1-\frac{1}{c_1(\bgm, k)}\right)>1-k.
\end{equation}
\item \label{2:lem:prel}
Let $\bgm \supsetneq T$. 
Then $ (1-R_T^2)/(1-R_{\bgm}^2)\geq 1$. Further 
for any $0<k<1$ and any positive integer $n$, 
there exists a $c_2(\bgm,T,k)>0 $ such that
\begin{equation}
 \mathrm{Pr}\left(1 \leq
\left(\frac{1-R_T^2}{1-R_{\bgm}^2}\right)^n<
1+c_2(\bgm,T,k)\right)>1-k.
\end{equation}
\item \label{3:lem:prel}
Let $\bgm \nsupseteq T$. Then 
for any $0<k<1$ and any positive integer $n$, 
there exists a $c_3(\bgm,T,k)>1$ such that
\begin{equation}
\mathrm{Pr}\left(\frac{1-R_T^2}{1-R_{\bgm}^2}<1-\frac{1}{c_3(\bgm,T,k)} \right)
>1-k.
\end{equation}
\end{enumerate}
\end{lemma}
\begin{proof}
 See Appendix \ref{app:prel}.
\end{proof}
First we give a consistency result on BIC.
\begin{thm}\label{thm:BICconsistency}
Assume A\ref{AS1} and A\ref{AS2}. 
The Bayes factor based on BIC under Gaussian errors
\begin{equation*}
\BF^G\gvf\BIC
= \left\{\frac{(1-R_{\bgm}^2)^{-n}n^{-q_{\bgm}}}{(1-R_F^2)^{-n}n^{-p}}\right\}^{1/2}
\end{equation*}
is consistent for model selection under \sss errors (including $\mcM_\bgm=\mcM_N$).
\end{thm}
\begin{proof}
We have only to show that
\begin{equation} \label{main-plim}
\plim_{n \to \infty}
\frac{\BF^G\gvf\BIC}{\BF^G_{T:F}\BIC}=
\plim_{n \to \infty}
\left\{ n^{q_T-q_{\bgm}} 
\left(\frac{1-R^2_T}{1-R^2_{\bgm}}\right)^{n}
\right\}^{1/2}=0.
\end{equation}
Consider the following two situations:
\begin{enumerate}
\item 
$\bgm \supsetneq T$: 
By part \ref{2:lem:prel} of Lemma \ref{lem:prel}, 
$\{(1-R^2_T)/(1-R^2_{\bgm})\}^{n}$ is bounded in probability.
Since $ q_{\bgm}> q_T$, \eqref{main-plim} is satisfied.
\item
$\bgm \nsupseteq T$: 
\quad
By part \ref{3:lem:prel} of Lemma \ref{lem:prel}, 
$ (1-R^2_T)/(1-R^2_{\bgm})$ 
is strictly less than $1$ in probability. 
Hence $ \{(1-R^2_T)/(1-R^2_{\bgm})\}^n$
converges to zero in probability
exponentially fast with respect to $n$.
Therefore, no matter what value $ q_T-q_{\bgm}$ takes,
\eqref{main-plim} is satisfied.
\end{enumerate}
These complete the proof.
\end{proof}
Note that in Theorem \ref{thm:BICconsistency}
 we do not exclude the null model $\mcM_N$ and hence
BIC has model selection consistency even when the null model is true.
When we consider consistency of the Bayes factors
treated in the previous sections, 
$\BF\gvf(\nu) $, $\BF^G\gvf[\nu,k] $, $ \widetilde{\BF}^G\gvf(\nu)$,
we have to exclude the null model $\mcM_N$, but they all still have
model selection consistency among non-null models.
\begin{corollary} \label{cor:consistency}
Assume A\ref{AS1} and A\ref{AS2}. $\{\nu,a\}$ is assumed independent of $n$ and
$\mcM_\bgm$. 
Assume also $\mcM_N$ is excluded from possible models.
Then 
\begin{enumerate}
\item $\widetilde{\BF}^G\gvf(\nu)$ for $ \nu<1$
is consistent for model selection under \sss errors.
\item $\BF^G\gvf[\nu,k]$ for $-k<\nu<1$ 
is consistent for model selection under \sss errors.
\item 
$\BF\gvf(\nu)$ for $0<\nu<1$ 
is consistent for model selection under \sss errors.
\end{enumerate}
\end{corollary}
\begin{proof}
By part \ref{1:lem:prel} of Lemma \ref{lem:prel}, when 
$ \mcM_\bgm\neq\mcM_N$, both $R^2_\bgm$ and $R^2_F$ 
are positive and strictly less than 1 
with probability 1. Hence both $\varphi(q_\bgm-\nu,1-R^2_\bgm)$ and
$\varphi(p-\nu,1-R^2_F)$ 
where
\begin{equation*}
  \varphi(s,r)= rs^{s-1}\{(1/r-1)e\}^{-s}
\end{equation*}
are positive and bounded from above
with probability 1 provided $\nu<1$ and $\nu$ is independent of $n$ and $\mcM_\bgm$.
(On the other hand, since $R_N^2 \equiv 0$, $\varphi(q_\bgm-\nu,1-R^2_N)$
is not defined.)
As in Theorem \ref{thm:Laplace}, 
\begin{equation*}
 \widetilde{\BF}^G\gvf(\nu)  =
\left\{
\frac{\varphi(q_{\bgm}-\nu,1-R_{\bgm}^2)}{\varphi(p-\nu,1-R_F^2)}\right\}^{1/2}
\BF^G\gvf\BIC .
\end{equation*} 
Hence consistency of $\widetilde{\BF}^G\gvf(\nu)$ follows from
consistency of BIC.

Further as $n\to\infty$, we have
\begin{equation*}
 \BF^G\gvf[\nu,k] \approx \widetilde{\BF}^G\gvf(\nu)
\end{equation*}
by Theorem \ref{thm:Laplace}
provided $-k<\nu<1$ and $\{\nu,k\}$ are independent of $n$ and $\mcM_\bgm$.
Hence consistency of $\BF^G\gvf[\nu,k]$ follows from
consistency of $\widetilde{\BF}^G\gvf(\nu)$.

Remember that $\BF\gvf(\nu)$, the Bayes factor w.r.t.~sub-harmonic
priors, under \sss errors, is equal to $\BF^G\gvf[\nu,0]$
for $0<\nu<1$. 
Hence consistency of $\BF\gvf(\nu)$ follows from consistency of $\BF^G\gvf[\nu,k]$.
\end{proof}
\begin{remark}
\cite{Liang-etal-2008} established model selection consistency 
for $\nu<0$ and $k=2-\nu$ for Gaussian errors. Corollary \ref{cor:consistency}
in conjunction with Theorem \ref{thm:Laplace}
extends their result to the entire class of \sss distributions for 
a broader class of $\nu$ and $k$.

It should be emphasized in each of the above cases that is the Bayes
factor method developed for the Gaussian case that is shown to have
model selection consistency for the entire class of \sss error distributions.
These Gaussian based Bayes factors, however, are not Bayes factors
for error distributions which are not Gaussian, the sole exception 
being our robust Bayes factors which are based on separable priors 
in the sense described earlier, and which are simultaneously 
(for all \sss distributions) Bayes factors relative to the same prior.
\end{remark}

\begin{remark}
The issue of model selection consistency in our setup, is somewhat
complicated by the wide choice of possible error distributions.
If all errors are normally distributed, then under our assumptions A\ref{AS2} 
on the design matrix $\bm{X}_F$, 
imply that each $R_{\bgm}^2$ approaches a constant,
and that $\|\bep\|^2/n \to 1$.
If on the other hand, all models are variance mixtures of Gaussians
with mixture variance distributed as a positive random variable $\xi$,
then $\|\bep\|^2/n \to \xi$ a random variable, and $R^2_{\bgm}$
also approaches a random variable which is bounded above and below
in probability provided that $\xi$ is similarly bounded.

In general philosophical terms, it might be better to assume that
the sequence of error terms $\bep=(\epsilon_1,\dots,\epsilon_n)'$ 
are exchangeable for all $n$.
By De finetti's Theorem, this would imply that the error terms
all have a variance mixture of normal distributions.
We have chosen a slightly weaker requirement on the sequence of error
distributions, namely, that $\|\bep\|^2/n$ remains bounded above and
below in probability, which extracts the necessary limiting behavior
of the error terms to ensure consistency of model selection.
Interestingly, although we attain model selection consistency with 
these assumptions, it is not necessarily true that
$1=\mathrm{var}\epsilon_i=\mathrm{var}\xi$ is consistently estimated
by $\|\bep\|^2/n$.
\end{remark}

\section{Including the null model}\label{sec:null}
We noted in the introduction that it is often important that the null
model be allowed as a possible model. In this section we show how our
method my be easily altered to include the null model.

First note for example that expression \eqref{main-BF} clearly 
shows why the null model is not allowed as a possibility in the development of that section.
For the null model $\mcM_N$, $R^2_N=0$ so the numerator of \eqref{main-BF} 
is infinite, and hence so would be $ \BF^G_{N:F}(\nu)$.
This situation may be avoided at a slight cost in complexity and
in interpretability of the expressions.
The required alteration in the prior distributions (proper and improper)
is to treat the intercept parameter $\alpha$ as another $\bbe$,
(and not give it a ``uniform'' prior).
This results in replacing the improper prior in \eqref{improper-joint} by
\begin{equation*} 
 \pigm(\check{\bbe}_{\bgm},\eta_{\bgm}|\nu)  = 
\frac{\Gamma(\frac{q_{\bgm}+1-\nu}{2})}{2^{\frac{\nu}{2}}\pi^{\frac{q_{\bgm}+1}{2}}}
|\check{\bm{X}}'_{\bgm}\check{\bm{X}}_{\bgm}|^{\frac{1}{2}} 
(\check{\bbe}'_{\bgm}\check{\bm{X}}'_{\bgm}\check{\bm{X}}_{\bgm}\check{\bbe}_{\bgm})
^{-\frac{q_{\bgm}+1-\nu}{2}}\eta_{\bgm}^{\frac{\nu}{2}-1},
\end{equation*}
where $\check{\bbe}_{\bgm}=(\alpha,\bbe'_{\bgm})'$ 
and $\check{\bm{X}}_{\bgm}=(\bm{1}_n | \bm{X}_{\bgm})$.
Similarly the marginal distribution in \eqref{MGG}
and the Bayes factor given by \eqref{main-BF}
are replaced by
\begin{equation*}
\check{M}^G_{\bgm}(\bm{y}|\nu)
 = \frac{\Gamma(n/2)}{\|\bm{y}\|^{n}\pi^{n/2}}
\int_0^{\infty} \frac{g^{\nu/2-1}(1+g)^{(n-q_{\bgm}-1)/2}}
{\left\{g(1-\check{R}_{\bgm}^2)+1 \right\}^{n/2}}\, dg, 
\end{equation*}
and
\begin{equation} \label{eq:BFcheck}
\check{\BF}^G\gvf(\nu)
=
\frac{\int_0^{\infty} g^{\frac{\nu}{2}-1}(1+g)^{\frac{n-q_{\bgm}-1}{2}}
\{g(1-\check{R}_{\bgm}^2)+1 \}^{-\frac{n}{2}}
\, dg}{\int_0^{\infty} g^{\frac{\nu}{2}-1}(1+g)^{\frac{n-p-1}{2}}
\{g(1-\check{R}_F^2)+1 \}^{-\frac{n}{2}}
\, dg},
\end{equation}
where 
\begin{equation*}
 \check{R}_{\bgm}^2
=1-\frac{\|\bm{Q}_\bgm(\bm{y}-\bar{y}\bm{1}_n)\|^2}{\|\bm{y}\|^2},\quad
 \check{R}_{F}^2
=1-\frac{\|\bm{Q}_F(\bm{y}-\bar{y}\bm{1}_n)\|^2}{\|\bm{y}\|^2},
\end{equation*}
(the ``coefficient of determination'' of the model $\mcM_{\bgm}$ relative to
the $0$-intercept model).
Hence with the substitution $R^2_{\bgm} \to \check{R}^2_{\bgm}$, 
$n-1 \to n$, $q_{\bgm} \to q_{\bgm}+1$, $\bm{y}-\bar{y}\bm{1}_n \to \bm{y}$, 
all expressions and results in the paper remain valid.
In particular,
the Bayes factor $\check{\BF}^G\gvf(\nu)$ is approximated by
\begin{equation}\label{eq:checkwidetildeBF}
 \widetilde{\check{\BF}}^G\gvf(\nu)  =
\left\{
\frac{\varphi(q_{\bgm}+1-\nu,1-\check{R}_{\bgm}^2)}{\varphi(p+1-\nu,1-\check{R}_F^2)}\right\}^{1/2}\BF^G\gvf\BIC 
\end{equation} 
since
\begin{equation*}
 \frac{1-\check{R}_{\bgm}^2}{1-\check{R}_{F}^2}=
 \frac{1-R_{\bgm}^2}{1-R_{F}^2}.
\end{equation*}
Since $\varphi(q_{\bgm}+1-\nu,1-\check{R}_{\bgm}^2)$ under the null model
is strictly positive and bounded,
\eqref{eq:checkwidetildeBF} is asymptotically equivalent to BIC with
a simple $O(1)$ rational correction function
depending upon $\nu$ as well as the $ \check{R}^2_{\bgm}$ and the numbers of predictors.
The approximation \eqref{eq:checkwidetildeBF} guarantees that
the result (Corollary \ref{cor:consistency}) on model selection consistency 
in Section \ref{sec:consistency} holds also for the null model.

In summary, the quantity $ \check{R}^2_{\bgm}$ is somewhat unusual, but
if model selection consistency under the null-model is desirable,
we can use $ \check{\BF}^G\gvf(\nu) $ or the approximated 
$ \widetilde{\check{\BF}}^G\gvf(\nu) $.
Interestingly, under the Gaussian regression setup, 
\cite{Guo-Speckman-2009} and \cite{Celeux-Anbari-Marin-Robert-2012}
recommend use of the Bayes factor as a function of $ \check{R}^2_{\bgm}$,
which just substitutes $g^{\nu/2-1}$ with $(1+g)^{-1}$ in 
$ \check{\BF}^G\gvf(\nu) $ given by \eqref{eq:BFcheck}.

\section{Examples}\label{sec:ex}
In this section, 
we provide illustrations of the method using both simulated and real data.
In each example, we compare several different versions of the Laplace
approximated Bayes factors $\widetilde{\BF}^G\gvf(\nu)$ and $\BF^G\gvf\BIC$.
The values of $\nu$ are $-2,-1,0,0.5,0.95$. These choices correspond to our default
choice, $\nu=1/2$ and $\nu=0.95$ which also satisfies our robustness condition $0<\nu<1$.
The choice $\nu=0$ approximates $\BF^G\gvf[0,2]$ of \cite{Guo-Speckman-2009} and
$\nu=-2$ and $\nu=-1$ approximates two choices of \cite{Liang-etal-2008}
as presented in Corollary \ref{cor:approxBF}.
\subsection{Simulation Studies}\label{sec:sim}
We compare numerical performance of our 
with BIC in a small simulation study.
We generated 16 possible correlated predictors ($p=16$) as follows:
\begin{equation*}
\begin{split}
& \overbrace{x_1, x_2}^{\mathrm{cor}=0.5},
\underbrace{x_3, x_4}_{\mathrm{cor}=-0.4},
 \overbrace{x_5, x_6}^{\mathrm{cor}=0.3},
 \underbrace{x_7, x_8}_{\mathrm{cor}=-0.2},
\overbrace{x_9, x_{10}}^{\mathrm{cor}=0.1}
\sim N(0,1) \\
& x_{11}, x_{12}, x_{13}, x_{14},
x_{15}, x_{16}  \sim N(0,1). 
\end{split}
\end{equation*}
Here ``cor'' denotes the correlation of two Gaussian random variables.
Also
$(x_1,x_2)$, $(x_3,x_4)$, $(x_5,x_6)$, $(x_7,x_8)$, $(x_9,x_{10})$,
$x_{11},x_{12}, x_{13}, x_{14}, x_{15},x_{16}$ are assumed to be independent.
After generating pseudo random $x_1,\dots,x_{16}$, we centered and scaled them
as noted in Section \ref{sec:intro}.
We set $n=30$ and consider 4 cases where the true predictors are
\begin{description}
 \item[$q_{T}=16$] \
$x_1, x_2, x_3, x_4, x_5, x_6, x_7, x_8, x_9, x_{10},
x_{11}, x_{12}, x_{13}, x_{14}, x_{15}, x_{16}$ 
 \item[$q_{T}=12$] \
$x_1, x_2, x_3, x_4, x_5, x_6, x_7, x_8,  x_{9}, x_{10}, x_{11}, x_{12}$ 
 \item[$q_{T}=8$] \ \, 
$x_1, x_2, \qquad \quad x_5, x_6, \qquad \quad x_9, x_{10}, x_{11} , x_{12}$ 
 \item[$q_{T}=4$] \  \,
 $x_1, x_2, \qquad \quad x_5, x_6$ 
\end{description}
(where $q_T$ denotes the number of true predictors) and the true model
is given by
\begin{equation} \label{simu-true-model}
\bm{y}=\bm{1}_{30} + 2 \sum_{i \in \{\mbox{true}\}}\bm{x}_i+ 
\sigma \times
\begin{cases}
 N_{30}(\bm{0},\bm{I}_{30}), \\ \mbox{Multi-}t(\bm{0},\bm{I}_{30};3,30),
\end{cases}
\end{equation}
with $\sigma=0.5,\ 1, \ 2$.
Tables \ref{normal} and \ref{multi-t} show 
how often the true model ranks first and how often 
it is in the top 3 among $2^{16}-1$ candidates
when the number of replicates is $N=200$. 
The error distributions are Gaussian (Table \ref{normal}) and multivariate-$t$
with 3 degrees of freedom (Table \ref{multi-t}).
For the case of normally distributed errors (Table \ref{normal}),
the Bayes factor methods performed well and stably for $\sigma=0.5$ and
$\sigma=1$ and did reasonably well for $\sigma=2$ for the smaller
true models ($q_T=4,8$). BIC seemed, generally, to have a preference
for larger models, and performed much less well than the Bayes factor method
for $\sigma=0.5$ and $\sigma=1$ for models of smaller size ($q_T=4,8,12$)
For $\sigma=2$, BIC did substantially better than BF for the largest
model ($q_T=16$) and somewhat better for $q_T=12$.
Performance of $\widetilde{\BF}^G\gvf(\nu)$ seemed relatively insensitive
to the choice of $\nu$. 
When $q_T \neq 16$, the choice of $\nu$ makes little difference. 
But when $q_T = 16$, positive $\nu=(0.5,0.95)$ seems
to perform better especially for larger $\sigma$.

Interestingly, for the case of a multivariate-$t$ error distribution with $3$
degrees of freedom (the minimum so that a variance exists),
the numerical results were quite similar to those in the normal case
for both $\widetilde{\BF}^G\gvf(\nu)$ and BIC, both quantitatively and qualitatively.
One possible aspect of the relative insensitivity of the results to 
choice of $\nu$ in heavy tailed case is the extension of model selection
consistency for the entire class of \sss errors to a broad class 
mixture of $g$-prior based methods given by Corollary \ref{cor:consistency}.

\begin{table} 
\setlength{\tabcolsep}{4pt}
\caption{Frequency of the true model (Gaussian error)}
\begin{center}
 \begin{tabular}{c|cc|cc|cc|cc} \toprule
$q_T$ & \multicolumn{2}{c}{16} & \multicolumn{2}{c}{12}
& \multicolumn{2}{c}{8} & \multicolumn{2}{c}{4} \\ \midrule 
rank & 1 & 1-3 & 1 & 1-3 &  1 & 1-3 &  1 & 1-3  \\
   \midrule \midrule
\multicolumn{3}{c}{$\sigma=0.5$} & \multicolumn{6}{c}{} \\ 
\cmidrule(r){1-4} 
$\widetilde{\BF}^G(0.95)$ & 1.00 & 1.00 & 0.94 &  1.00 & 0.94 & 0.99 & 0.89 & 0.99  \\
$\widetilde{\BF}^G(0.5)$ & 1.00 & 1.00 & 0.95 & 1.00 & 0.94 & 0.99 & 0.88 & 0.99 \\
$\widetilde{\BF}^G(0)$  & 1.00 & 1.00 & 0.95 & 1.00 & 0.94 & 0.99 & 0.88 & 0.99   \\
$\widetilde{\BF}^G(-1)$ & 1.00 & 1.00 & 0.96 & 1.00 & 0.94 & 1.00 & 0.89 & 0.99  \\ 
$\widetilde{\BF}^G(-2)$ & 1.00 & 1.00 & 0.96 & 1.00 & 0.94 & 1.00 & 0.89 & 0.99  \\ 
\cmidrule(r){1-1} 
BIC & 1.00 & 1.00 & 0.44 & 0.62 & 0.28 & 0.46 & 0.20 & 0.35  \\
\midrule \midrule
\multicolumn{3}{c}{$\sigma=1$} & \multicolumn{6}{c}{} \\ 
\cmidrule(r){1-4} 
$\widetilde{\BF}^G(0.95)$  & 0.85 & 0.92 & 0.87 & 0.99 & 0.86 & 0.97 & 0.76 & 0.95  \\
$\widetilde{\BF}^G(0.5)$  & 0.83 & 0.90 & 0.88 & 0.99 & 0.87 & 0.97 & 0.76 & 0.94 \\
$\widetilde{\BF}^G(0)$  & 0.80 & 0.87 & 0.89 & 1.00 & 0.87 & 0.97 & 0.75 & 0.94  \\
$\widetilde{\BF}^G(-1)$  & 0.73 & 0.81 & 0.89 & 1.00 & 0.88 & 0.97 & 0.74 & 0.94  \\
$\widetilde{\BF}^G(-2)$  & 0.55 & 0.72 & 0.90 & 1.00 & 0.89 & 0.98 & 0.76 & 0.95  \\
\cmidrule(r){1-1} 
BIC & 1.00 & 1.00 & 0.44 & 0.62 & 0.28 & 0.46 & 0.20 & 0.35  \\
\midrule \midrule
\multicolumn{3}{c}{$\sigma=2$} & \multicolumn{6}{c}{} \\ 
\cmidrule(r){1-4} 
$\widetilde{\BF}^G(0.95)$ & 0.06 & 0.11 & 0.26 & 0.42 & 0.50 & 0.74 & 0.51 & 0.73 \\
$\widetilde{\BF}^G(0.5)$ & 0.05 & 0.10 & 0.25 & 0.41 & 0.51 & 0.74 & 0.50 & 0.72 \\
$\widetilde{\BF}^G(0)$ & 0.05 & 0.10 & 0.24 & 0.41 & 0.52 & 0.73 & 0.49 & 0.72 \\
$\widetilde{\BF}^G(-1)$ & 0.04 & 0.06 & 0.22 & 0.39 & 0.43 & 0.74 & 0.48 & 0.72 \\
 $\widetilde{\BF}^G(-2)$ & 0.02 & 0.03 & 0.17 & 0.32 & 0.47 & 0.72 & 0.43 & 0.71 \\
\cmidrule(r){1-1} 
BIC & 0.62 & 0.77 & 0.31 & 0.48 & 0.24 & 0.40 & 0.19 & 0.34 \\
\bottomrule
 \end{tabular}
\end{center}
\label{normal}
\end{table}%

\begin{table} 
\setlength{\tabcolsep}{4pt}
\caption{Frequency of the true model (multi-$t$ error)}
\begin{center}
 \begin{tabular}{c|cc|cc|cc|cc} \toprule
$q_T$ & \multicolumn{2}{c}{16} & \multicolumn{2}{c}{12}
& \multicolumn{2}{c}{8} & \multicolumn{2}{c}{4} \\ \midrule 
rank & 1 & 1-3 & 1 & 1-3 &  1 & 1-3 &  1 & 1-3  \\
   \midrule \midrule
\multicolumn{3}{c}{$\sigma=0.5$} & \multicolumn{6}{c}{} \\ 
\cmidrule(r){1-4} 
$\widetilde{\BF}^G(0.95)$ & 0.94 & 0.95 & 0.92 & 0.98 & 0.85 & 0.97 & 0.84 & 0.96 \\
$\widetilde{\BF}^G(0.5)$ & 0.94 & 0.95 & 0.92 & 0.98 & 0.86 & 0.97 & 0.84 & 0.96 \\
$\widetilde{\BF}^G(0)$  & 0.93 & 0.95 & 0.93 & 0.98 & 0.86 & 0.97 & 0.84 & 0.96 \\
$\widetilde{\BF}^G(-1)$ & 0.93 & 0.94 & 0.93 & 0.98 & 0.87 & 0.97 & 0.83 & 0.96 \\
$\widetilde{\BF}^G(-2)$ & 0.91 & 0.92 & 0.95 & 0.98 & 0.89 & 0.98 & 0.84 & 0.96 \\
\cmidrule(r){1-1} 
BIC & 0.98 & 0.99 & 0.47 & 0.63 & 0.28 & 0.45 & 0.26 & 0.38 \\
\midrule \midrule
\multicolumn{3}{c}{$\sigma=1$} & \multicolumn{6}{c}{} \\ 
\cmidrule(r){1-4} 
$\widetilde{\BF}^G(0.95)$  & 0.67 & 0.70 & 0.75 & 0.84 & 0.72 & 0.88 & 0.70 & 0.85 \\
$\widetilde{\BF}^G(0.5)$  & 0.64 & 0.68 & 0.76 & 0.84 & 0.71 & 0.88 & 0.70 & 0.85 \\
$\widetilde{\BF}^G(0)$  & 0.62 & 0.67 & 0.77 & 0.84 & 0.72 & 0.89 & 0.69 & 0.84 \\
$\widetilde{\BF}^G(-1)$  & 0.58 & 0.63 & 0.76 & 0.83 & 0.75 & 0.89 & 0.69 & 0.84 \\
$\widetilde{\BF}^G(-2)$  & 0.49 & 0.58 & 0.76 & 0.83 & 0.75 & 0.89 & 0.69 & 0.84 \\
\cmidrule(r){1-1} 
BIC & 0.89 & 0.93 & 0.45 & 0.60 & 0.27 & 0.44 & 0.25 & 0.37 \\
\midrule \midrule
\multicolumn{3}{c}{$\sigma=2$} & \multicolumn{6}{c}{} \\ 
\cmidrule(r){1-4} 
$\widetilde{\BF}^G(0.95)$ & 0.14 & 0.20 & 0.28 & 0.37 & 0.37 & 0.51 & 0.45 & 0.58 \\
$\widetilde{\BF}^G(0.5)$ & 0.14 & 0.18 & 0.29 & 0.36 & 0.37 & 0.51 & 0.45 & 0.57 \\
$\widetilde{\BF}^G(0)$ & 0.13 & 0.17 & 0.29 & 0.36 & 0.38 & 0.51 & 0.45 & 0.57 \\
$\widetilde{\BF}^G(-1)$ & 0.09 & 0.13 & 0.29 & 0.33 & 0.38 & 0.51 & 0.44 & 0.56 \\
$\widetilde{\BF}^G(-2)$ & 0.07 & 0.11 & 0.24 & 0.34 & 0.35 & 0.49 & 0.40 & 0.55 \\
\cmidrule(r){1-1} 
BIC & 0.47 & 0.59 &0.28 & 0.39 & 0.21 & 0.32 & 0.20 & 0.30 \\
\bottomrule
\end{tabular}
\end{center}
\label{multi-t}
\end{table}%

\subsection{Analysis of real data} \label{sec:real}
In this section, we apply our methods (approximate Bayes factor and BIC)
to Hald data set presented and analyzed in \cite{Casella-Moreno-2006}
and to the US Crime data set in \cite{Raftery-Madigan-Hoeting-1997}.
See those papers for detailed descriptions of the data sets.
Table \ref{table:hald} and \ref{table:uscrime} present posterior probabilities 
based on $\widetilde{\BF}^G(\nu)$ of the top three selected models 
(assuming equal prior probabilities on all models) for several different choices
of $\nu$ ($0.95,0.5,0,-1,-2$). 
BIC was also included in the study. In each case,
the first, second and third ranked choices based on $\widetilde{\BF}^G(\nu)$
were identical regardless of the choices of $\nu$.
Also in each case the top ranked submodel based on $\widetilde{\BF}^G(\nu)$
was regarded as reasonable in the earlier papers. As in the simulation study,
and as noted in several previous studies, BIC seems to choose bigger models.
In particular for the Hald data, the top choice $\{x_1,x_2\}$ agrees
with that of \cite{Casella-Moreno-2006} and also of \cite{Berger-Pericchi-1996}
and \cite{Draper-Smith-1998}.

For the US Crime data, our top ranked model agrees with that of the Occam's
window posterior in Table 2 of \cite{Raftery-Madigan-Hoeting-1997}.
Interestingly our second ranked model includes $x_{15}$ which does not
occur in any of \citeapos{Raftery-Madigan-Hoeting-1997} Occam's window model choices,
but which does occur in several models chosen by such classical methods
as Mallow's $C_p$, adjusted $R^2$, etc.~in their Table 1.

\begin{table}
\setlength{\tabcolsep}{4pt}
\caption{Hald data: posterior probabilities of top 3 selected models}
\begin{center}
\begin{tabular}{cccccccc} \toprule
& & & \multicolumn{5}{c}{$\widetilde{\BF}^G$} \\
&  & $\nu$ & 0.95 & 0.5 & 0 & -1 & -2 \\
1 & $\{1,2, \, \ \quad\}$ & & 0.66 & 0.63 & 0.61 & 0.57 & 0.54 \\
2 & $\{1,\qquad 4\}$ & & 0.16 & 0.17 & 0.17 & 0.18 & 0.20 \\
3 & $\{1,2,\quad 4\}$ & & 0.06 & 0.07 & 0.07 & 0.08 & 0.08 \\
\midrule
& & & BIC & \multicolumn{4}{c}{} \\
1 & $\{1,2, \, \ \quad\}$ & & 0.25 & \multicolumn{4}{c}{} \\
2 & $\{1,2,\quad 4\}$ & & 0.23 & \multicolumn{4}{c}{} \\
3 & $\{1,2,3, \ \ \}$ & & 0.23 & \multicolumn{4}{c}{} \\
\bottomrule
\end{tabular}
 \end{center}
\label{table:hald}
\end{table}

\begin{table}
\setlength{\tabcolsep}{4pt}
\caption{US crime data: posterior probabilities of top 3 selected models}
\begin{center}
\begin{tabular}{cccccccc} \toprule
& & & \multicolumn{5}{c}{$\widetilde{\BF}^G$} \\
& & $\nu$ & 0.95 & 0.5 & 0 & -1 & -2  \\
1 & $\{1,3,4, \quad 9,11, \quad \ 13,14 \quad \ \}$   & & 0.020 & 0.019 & 0.018 & 0.016  & 0.015 \\
2 & $\{1,3,4,\quad 9,11, \quad \ 13,14,15\}$ & & 0.018 & 0.018 & 0.017 & 0.015 & 0.014 \\
3 & $\{1,3,\quad 5,9,11, \quad \ 13,14 \quad \ \}$    & & 0.013 & 0.013 & 0.012 & 0.011 & 0.010 \\
\midrule
& & & BIC & \multicolumn{4}{c}{} \\
1 & $\{1,3,4,\quad 9,11, \quad \ 13,14,15\}$ & & 0.035 & \multicolumn{4}{c}{} \\
2 & $\{1,3,4, \quad 9,11, \quad \ 13,14 \quad \ \}$ & & 0.026 & \multicolumn{4}{c}{} \\
3 & $\{1,3,4,\quad 9,11, 12, 13,14,15\}$ & & 0.019 & \multicolumn{4}{c}{} \\
\bottomrule
\end{tabular}
\end{center}
\label{table:uscrime}
\end{table}

\section{Concluding remarks}\label{sec:cr}
Bayesian model selection for linear regression models with Gaussian errors
has been a popular area of study for some time. There is also a substantial
literature devoted to studying the extension of Stein-type shrinkage estimators
from models with Gaussian errors to those with general \sss
errors. In particular, it has long been observed that certain shrinkage
estimators which improve over the least squares (LS) estimator for
Gaussian models also improve over the LS estimator simultaneously for all
\sss error models 
(See for example, \cite{Cellier-et-al-1989}).
\cite{Maruyama-2003b} and \cite{Maru-Straw-2005} found, in addition, that
certain separable priors (in the sense described in Section \ref{sec:prior})
lead to generalized Bayes shrinkage estimators that do not depend on
the form of the underlying \sss distribution and that
also simultaneously improve on the LS estimator, sometimes dramatically so.
The original aim of this research was to see if similar separable priors
could be found that have this distributional robustness property in the
variable selection problem (and, that also perform well with regard to
model selection consistency and with regard to good MSE performance of the
estimators of $\bbe_\bgm$ in each of submodels).

The generalized Bayes priors developed in sections \ref{sec:prior} and
\ref{sec:marginal} turned out to satisfy our requirements and also to be closely 
related to other so called $g$-priors (or mixtures of $g$-priors) in the
literature (See e.g.~\cite{Liang-etal-2008,Guo-Speckman-2009}).
We demonstrated that our sub-harmonic priors are the only ones in this class
that are robust in our sense.

The expression of our Bayes factors, e.g.~\eqref{eq:main-thm-BF},
are relatively simple involving the ratio of two $1$-dimensional integrals.
To further simplify calculations we investigated Laplace approximations
to our Bayes factors, and more generally, to a collection of Bayes factors arising from
mixtures of $g$-priors that have recently appeared 
(See \cite{Liang-etal-2008, Guo-Speckman-2009}). 
We show in Section \ref{sec:Laplace} that in each case the Bayes factor can be approximated
as the Bayes factor for the Gaussian model based on BIC times a simple
rational function depending $q$, $\nu$ and the $R^2$ of the models.

Using these Laplace approximations we are able to establish model selection
consistency of our robust  procedure for the entire class of \sss distributions
and to extend the model consistency results of several earlier papers
for the Gaussian case to the entire class of \sss distributions.

A small simulation study and an analysis of the Hald data 
(See \cite{Casella-Moreno-2006}) 
and the US Crime data (See \cite{Raftery-Madigan-Hoeting-1997}) indicates that our method 
performs well. It gives results consistent with the results of the cited 
papers for the real data sets and performs comparably
and sometimes better than several of the mixture of $g$-prior methods.

\appendix
\section{Proof of Lemma \ref{lem:relationship-G}}
\label{app:relationship-G}
Under the submodel $\mcM_{\bgm}$, 
the conditional marginal density of $\bm{y}$ with respect to improper prior
$\eta_{\bgm}^{\nu/2-1}$ given $\alpha$ and $\bbe_{\bgm}$ is
\begin{equation}
\begin{split}
& M_{\bgm}(\bm{y}|\{\alpha,\bbe_{\bgm}\},\nu)=\int_0^\infty \eta_{\bgm}^{n/2} 
f_n\left(\eta_{\bgm}\|\bm{y}-\alpha \bm{1}_n -\bm{X}_{\bgm}\bbe_{\bgm}\|^2\right)
\eta_{\bgm}^{\nu/2-1} d\eta_{\bgm} \\
& = \|\bm{y}-\alpha \bm{1}_n -\bm{X}_{\bgm}\bbe_{\bgm}\|^{-n-\nu}
\int_0^\infty t^{\{n+\nu\}/2-1}f_n(t)dt  \\
& 
= \frac{\int_0^\infty t^{(n+\nu)/2-1}f_n(t)dt}
{\int_0^\infty t^{(n+\nu)/2-1}f_G(t)dt}
\int_0^\infty 
f_G\left(\eta\|\bm{y}-\alpha \bm{1}_n -\bm{X}_{\bgm}\bbe_{\bgm}\|^2\right)
\eta^{\{n+\nu\}/2-1} d\eta \\
&= \frac{E[\|\bep_{\bgm}\|^\nu]}{E[\|\bep_G\|^\nu]}
\int_0^\infty 
\eta_{\bgm}^{n/2}f_G\left(\eta_{\bgm}\|\bm{y}-\alpha \bm{1}_n -\bm{X}_{\bgm}\bbe_{\bgm}\|^2\right)
\eta_{\bgm}^{\nu/2-1} d\eta_{\bgm} \\
&=\frac{E[\|\bep_{\bgm}\|^\nu]}{E[\|\bep_G\|^\nu]}
M^G_{\bgm}(\bm{y}|\{\alpha,\bbe_{\bgm}\},\nu)
\end{split}
\end{equation}
where $f_G(t)=(2\pi)^{-n/2}\exp(-t/2)$,
provided 
\begin{equation} \label{moment-cond}
\int_0^\infty t^{(n+\nu)/2-1}f_n(t)dt<\infty \ \Leftrightarrow \ 
E[\|\bep_{\bgm}\|^\nu]<\infty.
\end{equation}
Therefore, we have 
\begin{equation}\label{eq:MG}
\begin{split}
& M_{\bgm}(\bm{y}|\nu) =
\iint M_{\bgm}(\bm{y}|\{\alpha,\bbe_{\bgm}\},\nu)
\pi(\alpha,\bbe_{\bgm})d\alpha d\bbe_{\bgm} \\
&\quad =\iint M^G_{\bgm}(\bm{y}|\{\alpha,\bbe_{\bgm}\},\nu)
\pi(\alpha,\bbe_{\bgm})d\alpha d\bbe_{\bgm}  
=\frac{E[\|\bep_{\bgm}\|^\nu]} {E[\|\bep_G\|^\nu]}
M^G_{\bgm}(\bm{y}|\nu). 
\end{split}
\end{equation}

\section{Proof of Lemma \ref{lem:uniqueness}}
\label{app:uniqueness}
Let $\eta=z_1$ and $\bbe'\bm{X}'\bm{X}\bbe=z_2$.
Note that the integral in \eqref{eq:mixturesg} is
\begin{equation}
\begin{split}
 h(z_1,z_2)&= \int_0^\infty \frac{g^{\nu/2-1}\tilde{\pi}(g)}{g^{q_{\bgm}}/2}
\exp\left(-\frac{z_1z_2}{2g}\right)dg \\
&=\int_0^\infty t^{(q_{\bgm}-\nu)/2-1}\tilde{\pi}(1/t)\exp\left(-t\frac{z_1z_2}{2}\right)dt
\end{split}
\end{equation}
and hence is a function of $z_1z_2$. 
Additionally, separability of $h(z_1,z_2)$ implies (essentially by Cauchy's functional equation) that 
\begin{equation}
 h(z_1,z_2)=c_1(z_1z_2)^{-c_2}
\end{equation}
for some constant $c_1$ and $c_2$ and any $z_1z_2>0$.

When $z_1z_2\to \infty$, we have
\begin{equation}
\lim_{z_1z_2\to\infty}\frac{(z_1z_2)^{(q_{\bgm}-\nu)/2}h(z_1,z_2)}
{\Gamma(\{q_{\bgm}-\nu\}/2)2^{(q_{\bgm}-\nu)/2}} =1
\end{equation}
from the same Tauberian theorem used in Remark \ref{rem:objectivity}, which implies that
$c_1$ and $c_2$ should be
\begin{equation}
 c_1=\Gamma(\{q_{\bgm}-\nu\}/2)2^{(q_{\bgm}-\nu)/2}, c_2=(q_{\bgm}-\nu)/2.
\end{equation}
The equation
\begin{equation}
\frac{(z_1z_2)^{(q_{\bgm}-\nu)/2}h(z_1,z_2)}
{\Gamma(\{q_{\bgm}-\nu\}/2)2^{(q_{\bgm}-\nu)/2}} =1
\end{equation}
for any $z_1z_2>0$ is equivalent to
\begin{equation}\label{eq:app_comp}
 \int_0^\infty \left(\tilde{\pi}(1/t)-1\right)
f(t;(q_{\bgm}-\nu)/2,z_1z_2)dt=0
\end{equation}
where $ f(t;(q_{\bgm}-\nu)/2,z_1z_2)$ is the probability density of
gamma distribution with shape parameter $(q_{\bgm}-\nu)/2$ 
and scale parameter $z_1z_2/2$.
By completeness of the Gamma distribution,
\eqref{eq:app_comp} holds if and only if 
\begin{equation}
 \tilde{\pi}(1/t)=1,
\end{equation}
which completes the proof.

\section{Proof of Theorem \ref{thm:Laplace}}
\label{app:Laplace}
Denote the left-hand side of \eqref{eq:thmLaplace1} by $H(n)$.
When approximating $H(n)$,
make the change of variables $\tau= \log g$.
See  \cite{Liang-etal-2008} for details.
With this transformation, the integral becomes
\begin{equation}\label{check-M}
H(n)
=\int_{-\infty}^\infty
\frac{e^{(\nu/2-1)\tau}(1+e^\tau)^{(n-q-1)/2}}
{(1+e^{-\tau})^{k/2}\left\{1+re^\tau \right\}^{(n-1)/2}} e^\tau \, d\tau ,
\end{equation}
where the extra $e^\tau$ comes from the Jacobian of the transformation.
Denote the logarithm of the integrand function in 
\eqref{check-M} by $h(\tau,n)$.
We have
\begin{equation}\label{deriv-h-of-tau}
 \frac{\partial}{\partial \tau}h(\tau,n) 
= \frac{1}{2}\frac{z}{1+z}
\left\{\frac{(n-1)(1-r)}{1+rz} +\frac{\nu+k}{z}-(q-\nu)\right\}, 
\end{equation}
where $z=e^\tau$.
Since $0<r<1$ and $\nu+k>0$,
the equation $ \{\partial/\partial \tau\}h(\tau,n)=0$ has 
the only one positive root $ \hat{z}=e^{\hat{\tau}} $.
It clearly satisfies
\begin{equation}
\lim_{n \to \infty} \frac{\hat{z}}{n} = \frac{1/r-1}{q-\nu}.
\end{equation}
Hence we have
\begin{equation}
\begin{split}
 e^{h(\hat{\tau},n)} 
&= \left\{\hat{z}^{\nu} 
(1+\hat{z})^{n-q-1}(1+\hat{z}^{-1})^{-k}(1+r\hat{z})^{-n+1}\right\}^{1/2} \\
& = \left\{ 
\frac{\hat{z}^{-q+\nu}}{r^{n-1}}
\left( 1+\frac{n/\hat{z}}{n}\right)^{n-q-1-k}
\left(1+\frac{n/(r\hat{z})}{n}\right)^{-n+1} \right\}^{1/2}
\\
& \approx 
\left\{\left(\frac{q-\nu}{n\{1/r-1\}}\right)^{q-\nu}r^{-n+1}
\exp\left(\{1-1/r\}n/\hat{z}\right)
\right\}^{1/2} \\
& = \left\{\left(\frac{q-\nu}{n\{1/r-1\}e}\right)^{q-\nu}r^{-n+1}
\right\}^{1/2}. 
\end{split}
\end{equation}
Similarly, as in \eqref{deriv-h-of-tau},
we have 
\begin{equation*}
\frac{\partial^2}{\partial \tau^2}h(\tau,n)
 =\frac{(\partial/\partial \tau)h(\tau,n)}{1+z}-\frac{z^2}{2(1+z)}
\left\{\frac{(n-1)(1-r)r}{(1+rz)^2} +\frac{\nu+k}{z^2}\right\}
\end{equation*}
and
\begin{equation}
\frac{\partial^2}{\partial \tau^2}h(\tau,n)|_{\tau=\hat{\tau}}
 \approx 
-\frac{\hat{z}}{2(1+\hat{z})}
\frac{(n-1)(1-r)r\hat{z}}{(1+r\hat{z})^2} \\
 \approx -\frac{n}{2\hat{z}}\frac{1-r}{r}
\approx -\frac{q-\nu}{2}.
\end{equation}
Therefore we have
\begin{equation} \label{approx}
\begin{split}
H(n) & \approx (2\pi)^{1/2} e^{h(\hat{\tau},n)}
\left(\{-\partial^2/\partial \tau^2\}h(\tau,n)|_{\tau=\hat{\tau}}\right)^{-1/2} \\
& \approx \left\{\frac{4\pi}{q-\nu}
\left(\frac{q-\nu}{n\{1/r-1\}e}\right)^{q-\nu}
r^{-n+1}\right\}^{1/2}
\end{split}
\end{equation}
as $n \to \infty$. Hence part \ref{thm:Laplace:part1}
of the theorem follows.

Since $\BF^G\gvf[\nu,k]$ in \eqref{BF-1} is given by the ratio of
such integrals, part  \ref{thm:Laplace:part2} of the theorem follows.
\section{Proof of Lemma \ref{lem:prel}}
\label{app:prel}
Let $\mcM_T$ be the true submodel 
$\bm{y}=\alpha_T \bm{1}_n+ \bm{X}_T\bbe_T+\sigma_T\bep$
where 
$ \bm{X}_T$ is the $n \times q_T$ true design matrix and
$\bbe_T$ is the true ($q_T\times 1$) coefficient vector.

For the submodel $\mcM_{\bgm}$,
$1-R_{\bgm}^2$ is given by 
$ \|\bm{Q}_{\bgm}(\bm{y}-\bar{y}\bm{1}_n)\|^2/\|\bm{y}-\bar{y}\bm{1}_n\|^2$
with 
$\bm{Q}_{\bgm}
=\bm{I}-\bm{X}_{\bgm}(\bm{X}'_{\bgm} \bm{X}_{\bgm})^{-1}\bm{X}'_{\bgm}$.
The numerator and denominator are rewritten as
\begin{equation}\label{numerator-R2}
 \begin{split}
&  \|\bm{Q}_{\bgm}(\bm{y}-\bar{y}\bm{1}_n)\|^2 =
\|\bm{Q}_{\bgm} \bm{X}_T\bbe_T+\sigma_T\bm{Q}_{\bgm}\check{\bep}\|^2 \\
& \quad = \bbe'_T\bm{X}'_T\bm{Q}_{\bgm} \bm{X}_T\bbe_T 
+ 2\sigma_T\bbe'_T\bm{X}'_T\bm{Q}_{\bgm}\bep+
\sigma_T^2\check{\bep}'\bm{Q}_{\bgm} \check{\bep}
 \end{split}
\end{equation}
where $\check{\bep}=\bep-\bar{\epsilon}\bm{1}_n$ and similarly 
\begin{equation*}
\|\bm{y}-\bar{y}\bm{1}_n\|^2= \bbe'_T\bm{X}'_T \bm{X}_T\bbe_T 
+ 2\sigma_T\bbe'_T\bm{X}'_T\bep+
\sigma_T^2\|\check{\bep}\|^2.
\end{equation*}
Since $ \check{\bep}'\bm{Q}_{\bgm} \check{\bep} 
\leq \|\check{\bep}\|^2$, 
$1-R^2_{\bgm}$ is bounded as
\begin{equation} \label{1-R2}
\begin{split}
& \frac{\bbe'_T\{\bm{X}'_T\bm{Q}_{\bgm} \bm{X}_T/n\}\bbe_T 
+ 2\sigma_T\bbe'_T\{\bm{X}'_T\bm{Q}_{\bgm}\bep/n\}
+\sigma_T^2W_{\bgm} V_{n}} 
{\bbe'_T\{\bm{X}'_T \bm{X}_T/n\}\bbe_T 
+ 2\sigma_T\bbe'_T\{\bm{X}'_T\bep/n\}+
\sigma_T^2W_{\bgm} V_{n}} \\
& \quad \leq 1-R^2_{\bgm} \leq 
\frac{\bbe'_T\{\bm{X}'_T\bm{Q}_{\bgm} \bm{X}_T/n\}\bbe_T 
+ 2\sigma_T\bbe'_T\{\bm{X}'_T\bm{Q}_{\bgm}\bep/n\}+ \sigma_T^2V_{n}} 
{\bbe'_T\{\bm{X}'_T \bm{X}_T/n\}\bbe_T 
+ 2\bbe'_T\{\bm{X}'_T\bep/n\}+ \sigma^2_T V_{n}}
 \end{split}
\end{equation}
where $V_n=\check{\bep}'\check{\bep}/n$ 
and 
$W_{\bgm} = \check{\bep}'\bm{Q}_{\bgm} \check{\bep}/\|\check{\bep}\|^2 $.
As shown in \cite{Kelker-1970}, 
$W_{\bgm} \sim Be(\{n-q_{\bgm}-1\}/2, q_{\bgm}/2)$ even if
$\bep$ has a general (not necessarily Gaussian) \sss distribution.
In \eqref{1-R2}, we have the following.
\begin{itemize}
\item Since $E[\bep]=\bm{0}$ and 
$\mathrm{Var}[\bep]=\bm{I}_n$,
$E[\bm{X}'_T\bep/n]=\bm{0}$ and
\begin{equation} \label{var-in-probability}
\begin{split}
\mathrm{var}\left(\bm{X}'_T\bep/n\right)=
n^{-1}\{\bm{X}'_T\bm{X}_T/n\}  \to \bm{0}.
\end{split}
\end{equation}
Therefore $ \bbe'_T\bm{X}'_T\bep/n $ approaches 
$0$ in probability.
\item 
When $\bgm \supseteq T$, $\bm{Q}_{\bgm} \bm{X}_T$ is a zero matrix.
When $\bgm \nsupseteq T$, 
$\bbe'_T\{\bm{X}'_T\bm{Q}_{\bgm}\bep/n\} \to 0 $ in probability
can be proved as \eqref{var-in-probability}.
\item By the assumption A\ref{AS2}, 
$ \bm{X}'_T\bm{X}_T/n - \bm{X}'_T\bm{Q}_{\bgm} \bm{X}_T/n$ is 
positive-definite for any $n$
and hence
\begin{equation*}
\bbe'_T\{\bm{X}'_T \bm{X}_T/n\}\bbe_T> 
\bbe'_T\{\bm{X}'_T \bm{Q}_{\bgm} \bm{X}_T/n\}\bbe_T, 
\mbox{ for }\bbe_T\neq\bm{0}.
\end{equation*}
\item $ W_{\bgm}$ converges to $1$ in probability. 
\item By the assumption A\ref{AS1} on $\bep'\bep/n$, 
$ V_n$ is also
bounded in probability from below and from above.
\end{itemize}
Combining these facts, we see $0<R_{\bgm}^2<1 $ with strict inequalities in probability.

\medskip

Since $\bm{Q}_{\bgm} \bm{X}_T=\bm{0}$ for $\bgm \supseteq T$ and
using \eqref{numerator-R2}, 
$ (1-R^2_T)/(1-R_{\bgm}^2)$ is given by 
$ \|\bm{Q}_T\check{\bep}\|^2/\|\bm{Q}_{\bgm}\check{\bep}\|^2$.
Further we easily have
\begin{equation*}
1 \leq \frac{1-R^2_T}{1-R_{\bgm}^2}= 
\frac{\|\bm{Q}_T\check{\bep}\|^2}{\|\bm{Q}_{\bgm}\check{\bep}\|^2}
\leq 
\frac{\|\check{\bep}\|^2}{\|\bm{Q}_{\bgm}\check{\bep}\|^2}
=\frac{1}{W_{\bgm}}.
\end{equation*}
Note $W_{\bgm} \sim Be(\{n-q_{\bgm}-1\}/2, q_{\bgm}/2)$ 
is distributed as $(1+\chi^2_{q_{\bgm}}/\chi^2_{n-q_{\bgm}-1})^{-1}$
where $ \chi^2_{n-q_{\bgm}-1}$ and $\chi^2_{q_{\bgm}}$ are independent.
Hence
\begin{equation*}
\begin{split}
&\left\{ 1+\chi^2_{q_{\bgm}}/\chi^2_{n-q_{\bgm}-1}\right\}^{-n} 
= \left\{ 1+\left\{n/\chi^2_{n-q_{\bgm}-1}\right\} 
\left\{\chi^2_{q_{\bgm}}/n\right\} \right\}^{-n} \\
& \qquad \sim \exp(-\chi^2_{q_{\bgm}}) \mbox{ as }n \to \infty
\end{split}
\end{equation*}
since $ \chi^2_{n-q_{\bgm}-1}/n \to 1$ in probability.
Therefore $ W_{\bgm}^{-n}$ is bounded in probability from above
and hence the theorem follows.

\medskip
$ (1-R^2_T)/(1-R_{\bgm}^2) $ is written as
\begin{equation}
\begin{split}
&  \frac{1-R^2_T}{1-R_{\bgm}^2} 
 = \frac{\sigma_T^2\|\bm{Q}_T\check{\bep}\|^2}
{\bbe'_T\bm{X}'_T\bm{Q}_{\bgm} \bm{X}_T\bbe_T + 2\sigma_T\bbe'_T\bm{X}'_T\bm{Q}_{\bgm}\bep+
\sigma_T^2\check{\bep}'\bm{Q}_{\bgm} \check{\bep}} \\
 & \leq \frac{\sigma_T^2\|\check{\bep}\|^2}
{\bbe'_T\bm{X}'_T\bm{Q}_{\bgm} \bm{X}_T\bbe_T + 2\sigma_T\bbe'_T\bm{X}'_T\bm{Q}_{\bgm}\bep+
\sigma_T^2\check{\bep}'\bm{Q}_{\bgm} \check{\bep}} \\
& = \left( \frac{\bbe'_T \{\bm{X}'_T\bm{Q}_{\bgm} \bm{X}_T/n\}\bbe_T 
+ 2\sigma_T\bbe'_T\{\bm{X}'_T\bm{Q}_{\bgm}\bep/n\}}{\sigma_T^2 V_n}+ W_{\bgm} \right)^{-1}.
\end{split}
\end{equation}
Clearly $W_{\bgm} \to 1$ in probability. Also since
$\bgm \nsupseteq T$, 
$\bbe'_T \{\bm{X}'_T\bm{Q}_{\bgm} \bm{X}_T/n\}\bbe_T >0 $ for any $n$.
Further as $ \{\bm{X}'_T\bm{Q}_{\bgm}\bep/n\} \to \bm{0}$ in probability, 
$ (1-R^2_T)/(1-R_{\bgm}^2) $ is strictly smaller than $1$ in probability.

\end{document}